\documentclass[a4paper,reqno]{amsart} 

\usepackage{graphicx}
\usepackage{amsmath,amssymb,amsfonts,latexsym,amsthm}
\usepackage{url}
\usepackage{dsfont}
\usepackage{caption}
\usepackage{subcaption}
\usepackage{microtype}
\usepackage{fancyhdr}
\usepackage{textcomp}
\usepackage[active]{srcltx}

\usepackage[round]{natbib}

\usepackage[top=3cm, left=3cm, right=3cm]{geometry}
\usepackage[pdftex,colorlinks=true,urlcolor=blue,citecolor=blue, linkcolor=blue]{hyperref}

\newtheorem{theorem}{Theorem}[section]
\newtheorem{lemma}[theorem]{Lemma}

\newtheorem{proposition}[theorem]{Proposition}
\newtheorem{definition}[theorem]{Definition}

\newtheorem{remark}[theorem]{Remark}

\newtheorem{assumption}[theorem]{Assumption}

\numberwithin{equation}{section}

\newcommand{\F}{\mathcal{F}}

\newcommand{\Pred}{\mathcal{P}}

\newcommand{\filt}{\mathbb{F}}
\newcommand{\essinf}{\mathop{\mbox{ess inf}}}
\newcommand{\esssup}{\mathop{\mbox{ess sup}}}
\newcommand{\argmax}{\mathop{\mbox{argmax}}}
\newcommand{\argmin}{\mathop{\mbox{argmin}}}

\newcommand\abs[1]{\left\vert {#1} \right\vert}

\newcommand {\R}{\mathbb {R}}
\newcommand {\N}{\mathbb {N}}

\makeatletter

\date{}
\title[Good-deal hedging and valuation under combined uncertainty]{Good deal hedging and valuation under combined uncertainty about drift and volatility}

\author[D. Becherer]{Dirk Becherer}
\address[D. Becherer]{Institut f\"ur Mathematik, Humboldt-Universit\"at zu Berlin, D-10099 Berlin, Germany}
\email{becherer\,@\,mathematik.hu-berlin.de}
\author[K. Kentia]{Klebert Kentia}
\address[K. Kentia]{Institut f\"ur Mathematik, Goethe-Universit\"at Frankfurt, D-60054 Frankfurt a.M., Germany}
\email{kentia\,@\,math.uni-frankfurt.de}
\thanks{We would like to thank DFG, Berlin Mathematical School and RTG 1845 for support, and Xiaolu Tan for discussions and helpful suggestions.}

\begin{document}


\begin{abstract}
We study robust notions of good-deal hedging and valuation under combined uncertainty about the drifts and volatilities of asset prices.  
Good-deal bounds are determined by a subset of  risk-neutral pricing measures
 such that not only opportunities for arbitrage are excluded but also deals that are too good, by restricting instantaneous Sharpe ratios. 
A non-dominated multiple priors approach to model uncertainty (ambiguity) leads to worst-case good-deal bounds.
Corresponding hedging strategies arise as minimizers of a suitable coherent risk measure.
Good-deal bounds and hedges for measurable claims are characterized by solutions to second-order backward stochastic differential equations 
whose generators are non-convex in the volatility. 
These hedging strategies are robust with respect to uncertainty in the sense that their tracking errors satisfy a supermartingale property under all a-priori 
valuation measures, uniformly over all priors. 
\end{abstract}
\maketitle

\section{Introduction}
Hedging and valuation under model uncertainty (ambiguity) about volatility has been a seminal problem in the topical area of robust finance. In mathematics, it has motivated to no small extend recent advances on subjects
such as second order backward stochastic differential equations (2BSDEs), G-expectations and related stochastic calculus,  sub-linear conditional expectations and control of non-linear kernels,  using a variety of different methods from stochastic control, quasi-sure analysis and capacity theory, 
or expectation-spaces and PDE-theory, see \citet{ DenisMartini, DenisHuPeng, Nutz12, NutzSoner,SonerTouziZhang-Wellposedness,NutzHandel13, PengHuetal14,PossamaiTanZhou}
and many more references therein.
   The research has been challenging (and fruitful) since one has to deal (in a probabilistic setup) with families of non-dominated probability measures, also called multiple priors, that  can be mutually singular. In contrast, uncertainty solely 
   about drifts in a continuous time setting of stochastic It\^o-processes could be dealt with in a dominated framework of measures which are absolutely continuous with respect to a single reference probability measure.  

The main contributions of the current paper are twofold. For incomplete markets 
in continuous time, we solve the problem of robust  hedging and valuation under \emph{combined uncertainty about both the drifts and the volatilities} of the It\^o processes 
which describe the evolution of the tradeable asset prices  in an  underlying non-Markovian model for the financial market. 
Further, we investigate to this end the \emph{no-good-deal approach} to hedging and valuation, that is much cited in the finance literature  (cf. \citet{CochraneRequejo,CernyHodges,BjorkSlinko}) and provides more narrow valuation bounds and less extreme hedges than the more fundamental 
approach of almost-sure-hedging by  superreplication with its corresponding no-arbitrage valuation bounds.

Concerning the influential application of hedging under volatility uncertainty in continuous-time models which has been stipulated, at least, by \citet{AvellanedaetAl,Lyons95}, 
the  literature so far has almost entirely been concerned with the superreplication approach, for uncertainty being restricted solely to volatility as ambiguity about the 
drift of the asset prices has no effect there. 
While the notion of superreplication is fundamental to the theory of stochastic processes and for applications, being related to the optional decomposition and excluding the 
possibility of losses, it is also known from a practical point of view that superreplication could be overly expensive already in the absence of uncertainty for
 incomplete market models. 
 This calls for an adaption of other, less conservative, concepts for partial (not almost-sure) hedging in incomplete markets to a framework that is robust with respect to model ambiguity. 
Under combined uncertainties on drift and volatility, which are going to be relevant to such alternative approaches where the (ambiguous) distributions of hedging errors matter, however 
new mathematical challenges have to be overcome, as
the non-dominated family of measures will not just consist of local martingale measures, as e.g.\ in \citet{SonerTouziZhang-Wellposedness}.
Such has been noted  and addressed just recently in \citet{Nutz12,EpsteinJi2014,PossamaiTanZhou}. Likewise, we are aware of only few
recent  articles on the related problem of expected utility maximization under uncertainty about both drifts and volatilities  \citet{TevzadzeetAl2013,BiaginiPinar,NeufeldNutzUtilityMax}, some of which achieve quite explicit results for models with specific parametric structure. 
Among many interesting contributions on utility optimization under only one type of uncertainty, see for instance \citet{ChenEpstein,Quenez04,Garlappi,Schied07,OksendalSulem14} 
for solely (dominated) uncertainty about drifts,  or \citet{MatoussiPZ,HuEtal14} for ambiguity solely about volatilities but not about drifts. For equilibrium prices of a 
representative agent under ambiguity about the volatility, see \citep[][Sect.3.3]{EpsteinJi2013}.
 To the best of our knowledge, there appear to be hardly any studies on hedging approaches for (generically) incomplete markets  under combined ambiguity about drifts and volatilities  
- apart from superreplication.
 
We are going to investigate a robust extension of the no-good-deal hedging approach in continuous time under combined ambiguity about the volatilities and drifts.
Without model uncertainty, good-deal bounds have been introduced as valuation bounds in incomplete markets which do not only prevent 
opportunities for arbitrage  but also for deals with an overly attractive risk-to-reward ratio.
The most cited reference in the finance literature appears to be  \citet{CochraneRequejo}. We refer 
to  \citet{BjorkSlinko,CernyHodges, KloppelSchweizer} for mathematical and conceptual ideas and many more references.
By using only a suitable subset of ``no-good-deal'' risk neutral prices, the resulting valuation bounds 
are tighter than the classical no-arbitrage bounds (which are often too wide) but still have economic meaning.
Good-deal bounds have been defined predominantly by constraints on the  instantaneous Sharpe ratios in (any) extension 
of the financial market by additional derivatives' price processes, see \citet{BjorkSlinko}. For model without jumps, such is equivalent to imposing constraints on the 
optimal expected growth rates, see \citet{Becherer-Good-Deals}.
 Although good-deal theory has been initiated merely as a valuation approach \citep[cf.\ the conclusions in][]{BjorkSlinko}, 
 a corresponding approach to  hedging has been proposed, cf.\  \citet{Becherer-Good-Deals}, where  (good-deal)  hedging strategies are defined as 
minimizers of a certain dynamic coherent risk measure, in the spirit of \citet{BarrieuElKaroui}, such that  the good-deal bounds appear as 
a market consistent risk measures.
Naturally, results on valuations and hedges in good-deal theory, like in \citet{CochraneRequejo,BjorkSlinko,Becherer-Good-Deals}, 
are sensitive to the assumptions of the probability model on the drifts and volatilities of the asset prices. 
Since the objective real world measure is not precisely known and financial models can, at best, be useful but idealized simplifications of reality, 
robust approaches to model ambiguity are relevant to good-deal theory.

As far as we know, a robust approach to good-deal hedging in continuous time under non-dominated 
uncertainty has not been available so far.  For drift uncertainty and more references see  \citet{BechererKentiaTonleu}. Robust results on valuation and hedging  
 will be obtained by  2BSDE theory, building on recent advances by \citet{PossamaiTanZhou} whose analysis  provides a general wellposedness result that fits well for the present 
 application under combined uncertainty, cf.\ Remark \ref{rem:DynProgPrinciple}.
 Indeed, their result neither requires convexity nor uniform continuity of the generator, and it can deal with general (measurable) contingent claims without assuming further 
 regularity (like e.g.\ uniform continuity) or a Markovian framework.

The organization of the present paper is as follows. The setup and preliminaries are explained in
Section \ref{sec:Preliminaries}, with a brief summary of key results on 2BSDEs. 
Then we begin Section \ref{sec:NGDRestrictionandFinMarketVolUncer} by a description of the financial market and the (non-dominated) confidence set of (uncertain) priors that 
captures the combined ambiguity about  drifts and volatilities. Let us note that, in comparison to most literature on hedging under ambiguous volatility, 
like \citet{AvellanedaetAl,Lyons95}, we are going to consider a model for asset prices that constitutes a 
generically \emph{incomplete market}, even if seen just under (any) one individual prior. That means not only that there exists in general  
no perfect hedging (i.e.\ replicating) strategy which is robust with respect to ambiguity on priors, but that there does not even exists a replicating strategy in general in 
the model for (any) one given probability prior, without ambiguity. Section \ref{sec:NGDRestrictionandFinMarketVolUncer} then proceeds by taking drift and volatility  to be known  at first, in order to explain the idea for the good-deal approach as simply as possible. Following classical good-deal theory, good-deal 
restrictions are defined by constraints on the instantaneous Shape ratios, i.e.\ by radial bounds on the Girsanov kernels of pricing measures,
and standard BSDE descriptions of valuation bounds and hedges are summarized. 
Adopting a multiple-priors approach, like e.g.\ in \citet{GilboaS89,ChenEpstein}, while accommodating 
for the fact that priors here are non-dominated, Section \ref{sec:GoodDealValandHedgingwithVolUncer} starts by defining the good-deal bounds under uncertainty as 
the worst-case bounds over all priors. For hedging purposes, we define good-deal hedging strategies 
as minimizers of suitable a-priori risk measures under optimal risk sharing with the market. We derive  2BSDE characterizations 
for the dynamic valuation bounds and the hedging strategies. We show that 
tracking errors from  good-deal hedging satisfy a supermartingale property under all a-priori valuation measures uniformly for all priors. 
The proof relies on saddle-point arguments to identify the robust good-deal hedging strategy through a 
minmax identity. Finally, we finish Section \ref{sec:GoodDealValandHedgingwithVolUncer} 
with a simple but instructive  example about hedging a put option on a non-traded (but correlated) asset in an incomplete market. 
 This  allows for an elementary closed form solution, offering  intuition for the general but abstract main Theorem~\ref{thm:GDHedgingTheorem}.
It illustrates for instance that the good-deal hedging strategy generally is very different from the super-replicating strategy, which has been studied in e.g.\ 
 \citet{AvellanedaetAl,Lyons95,DenisMartini,NutzSoner,NeufeldNutz,Vorbrink2014}. The concrete case study also illustrates, how already in an elementary Markovian example 
 additional complications arise from \emph{combined} uncertainty.

\section{Mathematical framework and preliminaries}\label{sec:Preliminaries}
We consider filtered probability space $(\Omega,\F=\F_T,P^0,\filt)$ where $\Omega$ is the canonical space $\left\lbrace\omega\in \mathcal{C}([0,T],\R^n):\omega(0)=0\right\rbrace$ of continuous 
paths starting at $0$ endowed with the norm $\Vert \omega\Vert_\infty := \sup_{t\in[0,T]}\abs{\omega(t)}$. The filtration $\filt =(\F_t)_{t\in[0,T]}$ is generated by the 
canonical process $B_t(\omega):=\omega(t)$, $\omega\in\Omega$ and $P^0$ is the Wiener measure. 
We denote by $\filt_+= (\F_t^+)_{t\in[0,T]}$ the right-limit of $\filt$, with $\F^+_t=\F_{t+}:=\cap_{s>t}\F_s$. For a probability measure $Q$, the conditional expectation 
given $\F_t$ will be denoted by $E^Q_t[\cdot]$.
A probability measure $P$ is called a local martingale measure if $B$ is a local martingale w.r.t.\ $(\filt, P)$.
One can, cf.\   \citet{Karandikar}, construct the quadratic variation process 
$\left\langle B\right\rangle$ pathwise such that it coincides with $\langle B\rangle^P$ $P$-a.s.\ for all local martingale measures $P$. 
In particular this yields a pathwise definition of the density  $\widehat{a}$ of $\langle B\rangle$ w.r.t.\ the Lebesgue measure as
\begin{equation*}
 \widehat{a}_t(\omega) :=\limsup_{\epsilon\searrow0}\frac{1}{\epsilon}\big(\langle B\rangle_t(\omega) - \langle B\rangle_{t-\epsilon}(\omega)\big),\quad (t,\omega)\in[0,T]\times \Omega.
\end{equation*}
We denote by $\overline{\Pred}_W$ the set of all local martingale measures $P$ for which $\widehat{a}$ is well-defined and takes values $P$-almost surely in the space $\mathbb{S}^{>0}_n\subset \R^{n\times n}$ of positive 
definite symmetric $n\times n$-matrices. Note that the measures in $\overline{\Pred}_W$ can be mutually singular, as illustrated e.g. in \citet{SonerTouziZhang-Aggregation}. 
For any $P\in \overline{\Pred}_W$, the process $W^P:={\phantom{}}^{(P)}\hspace{-0.1cm}\int_0^\cdot\widehat{a}^{-\frac{1}{2}}_sdB_s$ is a Brownian motion under $P$. 
To formulate volatility uncertainty, we concentrate only on 
the subclass $\overline{\Pred}_S\subset \overline{\Pred}_W$ of measures
\begin{equation*}
 P^\alpha :=P^0\circ(X^\alpha)^{-1},\quad \text{where } X^\alpha:=\phantom{}^{\phantom{}^{(P^0)}}\hspace{-0.15cm}\int_0^\cdot\alpha^{1/2}_sdB_s,
\end{equation*}
with $\mathbb{S}^{>0}_n-$valued $\filt$-progressive $\alpha$ satisfying $\int_0^T\lvert\alpha_t\rvert dt<\infty,\ P^0\text{-a.s.}$. 
A benefit of restricting to the subclass $\overline{\Pred}_S$ is the following aggregation property \citep[cf.][Lem.8.1, Lem.8.2]{SonerTouziZhang-Aggregation}.
\begin{lemma}\label{lem:QSAggregPS}
 For $P\in\overline{\mathcal{P}}_W$, let $\filt^P$  denote the $P$-augmentation of the filtration $\filt$ and $\overline{\filt^{W^P}}^P$ that of the natural filtration $\filt^{W^P}$ of $W^P$. Then
 $B$ has the martingale representation property  w.r.t.\  $(\filt^P,P)$ for all  $P\in \overline{\Pred}_S$, 
and  $ \overline{\Pred}_S = \big\lbrace P\in \overline{\Pred}_W:\ \filt^P = \overline{\filt^{W^P}}^P\big\rbrace$. Moreover, $(P,\filt)$ satisfies the Blumenthal 
 zero-one law for any $P\in \overline{\Pred}_S$.
\end{lemma}
\begin{remark}
For any $P\in\overline{\Pred}_S$, Lemma \ref{lem:QSAggregPS} implies $E^P_t[X] = E^P[X\,\vert\,\F^+_t]$ $P\text{-a.s.}$ for any $X$ in $L^1(P)$, $t\in[0,T] $. In particular, any
$\F_t^+$-measurable random variable has a $\F_t$-measurable $P$-version.
\end{remark}
Let $\underline{a},\overline{a}\in \mathbb{S}^{>0}_n$. We will work with the subclass $\Pred_{[\underline{a},\overline{a}]}$ of $\overline{\Pred}_S$ defined by 
\begin{equation}\label{eq:DefPH}
 \Pred_{[\underline{a},\overline{a}]} := \left\lbrace P\in \overline{\Pred}_S:\ \underline{a}\le\widehat{a}\le\overline{a},\ P\otimes dt\text{-a.e.}\right\rbrace
\end{equation}
and assumed to be non-empty. We use  
the language of quasi-sure analysis as it appears in the framework of capacities of \citet{DenisMartini} as follows.
\begin{definition}
 A property is said to hold $\mathcal{Q}$-quasi-everywhere ($\mathcal{Q}$-q.e.\ for short) for a family $\mathcal{Q}$ of measures on the same measurable space if it holds outside of a set, which is 
 a nullset under each element of $\mathcal{Q}$. 
\end{definition}
Unless stated otherwise, inequalities between random variables will be meant in a $\Pred_{[\underline{a},\overline{a}]}$-quasi-sure sense (written $\Pred_{[\underline{a},\overline{a}]}$-q.s.\ for short), while inequalities 
between $\filt_+$-progressive 
processes will be in the $\Pred_{[\underline{a},\overline{a}]}\otimes dt$-q.e.\ sense,  for $\Pred_{[\underline{a},\overline{a}]}\otimes dt := \big\lbrace P\otimes dt,\ P\in\Pred_{[\underline{a},\overline{a}]}\big\rbrace$. 
We now introduce spaces and norms of interest for the paper. Some of these spaces are already quite classical, and have been modified here to account for the possible mutual 
singularity of measures in $\Pred_{[\underline{a},\overline{a}]}$. For a filtration $\mathbb{X}=(\mathcal{X}_t)_{t\in[0,T]}$ on 
$(\Omega,\F_T)$ with augmentation $\mathbb{X}^P:=\big(\mathcal{X}^P_t\big)_{t\in[0,T]}$  under  measure $P\in\Pred_{[\underline{a},\overline{a}]}$, we consider the following function 
spaces:
 \\
 a) $L^2_{\Pred_{[\underline{a},\overline{a}]}}(\mathcal{X}_T)$ (resp. $L^2(\mathcal{X}_T,P)$) of $\mathcal{X}_T-$measurable real-valued random variables $X$ with norm
			    \(
			    \lVert X\rVert^2_{L^2_{\Pred_{[\underline{a},\overline{a}]}}}=\sup_{P\in\Pred_{[\underline{a},\overline{a}]}}E^P\big[\abs{X}^2\big]<\infty\
			    \)
			    \(\left(\text{resp. } \lVert X\rVert^2_{L^2(P)}=E^P\big[\abs{X}^2\big]<\infty\right),
			    \) 
\\
	 b) $\mathbb{H}^2(\mathbb{X})$ (resp. $\mathbb{H}^2(\mathbb{X},P)$) of $\mathbb{X}-$predictable $\R^n-$valued processes $Z$ with 
 \begin{equation*}
 \lVert Z\rVert^2_{\mathbb{H}^2}=\sup_{P\in\Pred_{[\underline{a},\overline{a}]}}E^P\Big[{\int_0^T}\big\vert\widehat{a}_t^{\frac{1}{2}}Z_t\big\vert^2dt\Big]\hspace{-0.13cm}<\hspace{-0.1cm}\infty\
 	\Big(\text{resp.}\ \lVert Z\rVert^2_{\mathbb{H}^2(P)}=E^P\Big[{\int_0^T}\big\vert\widehat{a}_t^{\frac{1}{2}}Z_t\big\vert^2dt\Big]\hspace{-0.12cm}<\hspace{-0.1cm}\infty\Big),
  \end{equation*}
\\
c)
 $\mathbb{D}^2(\mathbb{X})$ (resp.\ $\mathbb{D}^2(\mathbb{X},P)$)  of all $\mathbb{X}-$progressive $\R$-valued processes $Y$ with   
  c\`adl\`ag paths $\Pred_{[\underline{a},\overline{a}]}$-q.s.\  (resp.\ $P$-a.s.), and satisfying 
	  \begin{equation*}\lVert Y\rVert_{\mathbb{D}^2}:=\Big\lVert \sup_{t\in[0,T] }\vert Y_t\vert\Big\rVert_{L^2_{\Pred_{[\underline{a},\overline{a}]}}}<\infty\ \Big(\text{resp.\ } \lVert Y\rVert_{\mathbb{D}^2(P)}:=\Big\lVert \sup_{t\in[0,T] }\vert Y_t\vert\Big \rVert_{L^2(P)}<\infty\Big),\end{equation*}
\\
d)  $\mathbb{L}^2(\mathbb{X})$ the subspace of $L^2_{\Pred_{[\underline{a},\overline{a}]}}(\mathcal{X}_T)$ consisting of random variables $X$ satisfying 
	  \begin{equation*}\lVert X\rVert^2_{\mathbb{L}^2}:=\sup_{P\in\Pred_{[\underline{a},\overline{a}]}}E^P\Big[{{\esssup^{\qquad\quad P}_{t\in[0,T] }}}\esssup^{\qquad\ \quad P}_{P'\in\Pred_{[\underline{a},\overline{a}]}(t,P,\mathbb{X})}E^{P'}\big[\lvert X\rvert^2\,\big\vert\, \mathcal{X}_t\big]\Big]<\infty,\end{equation*}
 for the set of measures
 $\Pred_{[\underline{a},\overline{a}]}(t,P,\mathbb{X}) := \left\lbrace P'\in\Pred_{[\underline{a},\overline{a}]}: P'=P\text{ on }\mathcal{X}_t\right\rbrace,$
\\
e) 
$\mathbb{I}^2(\mathbb{X},P)$ the space of $\mathbb{X}$-predictable processes $K$ with c\`adl\`ag and non-decreasing paths $P$-a.s., $K_0=0\ P\text{-a.s.}$, and 
 $\lVert K\rVert^2_{\mathbb{I}^2(P)} := E^P[K^2_T]<\infty.$ In particular we denote by $\mathbb{I}^2\Big(\big(\mathbb{X}^P\big)_{P\in\Pred_{[\underline{a},\overline{a}]}}\Big)$
 the family of tuples $(K^P)_{P\in \Pred_{[\underline{a},\overline{a}]}}$ s.t.\ $K^P\in \mathbb{I}^2(\mathbb{X}^P,P)$ for any $P\in \Pred_{[\underline{a},\overline{a}]}$
 and $\sup_{P\in \Pred_{[\underline{a},\overline{a}]}}\lVert K^P\rVert_{\mathbb{I}^2(P)}<\infty.$

The reader will note the analogy  with the spaces and norms defined in \citet{PossamaiTanZhou} (though with slightly different notations) for the specific family of collection of measures $\Pred(t,\omega):=\Pred_{[\underline{a},\overline{a}]}$ for 
any $(t,\omega)\in[0,T]\times\Omega$. A filtration which might in the sequel play the role of $\mathbb{X}$ in the definitions of spaces above is 
$\filt^{\Pred_{[\underline{a},\overline{a}]}}=\big(\F_t^{\Pred_{[\underline{a},\overline{a}]}}\big)_{t\in[0,T]},$ with 
$\F_t^{\Pred_{[\underline{a},\overline{a}]}}:=\bigcap_{P\in\Pred_{[\underline{a},\overline{a}]}}\F^P_t,\ t\in[0,T].$ 
\subsection{Second order backward stochastic differential equations}\label{sec:2BSDEs}
Following \citet{PossamaiTanZhou}, we summarize an existence and uniqueness result for Lipschitz 2BSDEs and state a 
representation of solutions that will be key to characterize the good-deal bounds and hedging strategies under combined drift and volatility uncertainties: See Proposition \ref{pro:ExistenceUniquenessSol2BSDEs}.
The generator for a 2BSDE is a function $F:[0,T]\times\Omega\times\R\times\R^n\times\mathbb{S}^{>0}_n\to \R$ for which we denote 
$\widehat{F}_t(\omega,y,z):=F_t(\omega_{\cdot\wedge t},y,z,\widehat{a}_t(\omega))$ and $\widehat{F}^0_t:=\widehat{F}_t(0,0)$. 
For wellposedness we will  require generators $F$ and terminal conditions $X$ that satisfy the following combination of Assumption 2.1.(i)-(ii) and Assumption 3.1.\ in \citet{PossamaiTanZhou}  (for $\kappa=p=2$).
\begin{assumption}\phantomsection{}\label{asp:Assumption1} 
 \begin{enumerate}
           \item [(i)]  $X$ is $\F_T$-measurable, 
           \item [(ii)] $F$ is jointly 
Borel measurable, and $\filt$-progressive in $(t,\omega)$ for each $(y,z,a)$,
           \item [(iii)] $\exists\, C>0$ such that for all $(t,\omega,a)\in[0,T]\times\Omega\times \mathbb{S}^{>0}_n,$ $y,y'\in\R,\ z,z'\in\R^n,$ 
           \begin{equation*}
            \big\lvert F_t(\omega,y,z,a) - F_t(\omega,y',z',a)\big\rvert\le C\Big(\lvert y-y'\rvert+\lvert z-z'\rvert\Big),
           \end{equation*}
	  \item [(iv)] $\widehat{F}^0$ satisfies $\Big(\int_0^T\lvert\widehat{F}^0_s\rvert^2ds\Big)^{1/2}\in \mathbb{L}^2(\filt_+)$.
 \end{enumerate}
\end{assumption}

\begin{remark}
Assumption \ref{asp:Assumption1}-(\text{iv}) is satisfied for $F$ such that Assumption \ref{asp:Assumption1}-(ii) holds and $\widehat{F}^0$ is bounded $\Pred_{[\underline{a},\overline{a}]}$-q.s..
It implies $\displaystyle \sup_{P\in\Pred_{[\underline{a},\overline{a}]}}E^P\Big[\int_0^T\lvert\widehat{F}^0_s\rvert^2ds\Big]<\infty.$ 
\end{remark}
A second-order BSDE is a stochastic integral equation of the type
\begin{equation}\label{eq:2BSDE}
 Y_t= X - \int_t^T\widehat{F}_s(Y_s,\widehat{a}^{1/2}_sZ_s)ds - \phantom{}^{\phantom{}^{(P)}}\hspace{-0.15cm}\int_t^TZ_s^{\text{tr}}dB_s + K^P_T-K^P_t, \ t\in[0,T] ,\ \Pred_{[\underline{a},\overline{a}]}\text{-q.s.}.
\end{equation}
In comparison to \citet{PossamaiTanZhou},  because the canonical process $B$ satisfies the martingale representation property simultaneously under all measures in $\Pred_{[\underline{a},\overline{a}]}$ (cf.\ Lemma \ref{lem:QSAggregPS}),
we do not have the orthogonal martingale components in the formulation of 2BSDEs as (\ref{eq:2BSDE}). The same formulation can be used in the more general framework with semimartingale laws for 
the canonical process, but working under a saturation property for the set of priors \citep[cf.][Def.5.1]{PossamaiTanZhou}.
\begin{definition}
 $(Y,Z,(K^P)_{P\in\Pred_{[\underline{a},\overline{a}]}})\in \mathbb{D}^2\big(\filt^{\Pred_{[\underline{a},\overline{a}]}}\big)\times \mathbb{H}^2\big(\filt^{\Pred_{[\underline{a},\overline{a}]}}\big)\times \mathbb{I}^2\big(\big(\filt^P\big)_{P\in\Pred_{[\underline{a},\overline{a}]}}\big)$ 
 is called solution (triple) to the 2BSDE (\ref{eq:2BSDE}) if it satisfies the required dynamics $\Pred_{[\underline{a},\overline{a}]}\text{-q.s.}$ and 
the family $\big\{K^P,\ P\in\Pred_{[\underline{a},\overline{a}]}\big\}$ satisfies the minimum condition 
\begin{equation}\label{eq:MinCond}
 K^P_t = \essinf^{\qquad\quad P}_{P'\in\Pred_{[\underline{a},\overline{a}]}(t,P,\filt_+)}E^{P'}_t[K^{P'}_T],\ t\in[0,T],\ P\text{-a.s.},\text{ for all } P\in \Pred_{[\underline{a},\overline{a}]}.
\end{equation}
If the family $\{K^P,\ P\in\Pred_{[\underline{a},\overline{a}]}\}$ can be 
aggregated 
\label{aggrnotion}
into a single process $K$, i.e.\ $K^P=K,\ P\text{-a.s.}$ for all $P\in\Pred_{[\underline{a},\overline{a}]}$, 
then $(Y,Z,K)$ is said to solve the 2BSDE.
\end{definition}
\begin{remark}\label{rem:OnAggregationOfK}	
	Note in the 2BSDE dynamics (\ref{eq:2BSDE}) the dependence of the stochastic integrals $\phantom{}^{(P)}\hspace{-0.1cm}\int_0^\cdot Z_s^{\text{tr}}dB_s$ on the probability measures 
	$P\in \Pred_{[\underline{a},\overline{a}]}$. Indeed since the measures in $\Pred_{[\underline{a},\overline{a}]}$ may be non-dominated, it might be that these integrals do not aggregate 
	\citep[see][for more on aggregation]{SonerTouziZhang-Aggregation}. However under additional set theoretical assumptions (for instance continuum hypothesis 
	plus the axiom of choice) 
	a method by \citet{Nutz} can be used to construct the stochastic integral $\int_0^\cdot Z_s^{\text{tr}}dB_s$ pathwise for any predictable process $Z$.
As a by-product, 
the family $\{K^P,\, P\in\Pred_{[\underline{a},\overline{a}]}\}$  for a 2BSDE solution $(Y,Z, (K^P)_{P\in\Pred_{[\underline{a},\overline{a}]}})$ would automatically aggregate into a 
single process $K:=Y_0-Y+\int_0^\cdot \widehat{F}_s(Y_s,\widehat{a}^{1/2}_sZ_s)ds+ \int_0^\cdot Z_s^{\text{tr}}dB_s$ yielding a 2BSDE solution $(Y,Z,K)$. 
Reciprocally for a solution $(Y,Z,K)$, the family $\big\lbrace\phantom{}^{(P)}\hspace{-0.1cm}\int_0^\cdot Z_s^{\text{tr}}dB_s,\ P\in \Pred_{[\underline{a},\overline{a}]}\big\rbrace$ 
also aggregates. 
\end{remark}
The pair $(F,X)$ will be called the parameters of the 2BSDE (\ref{eq:2BSDE}). We will refer to $Y$ as the value process and $Z$ as the control process. 
The following proposition provides the wellposedness result of interest in this paper, as well as a representation of the value process in terms of 
solutions to standard BSDEs.
The proof relies on an application of \citep[][Thm.4.1, Thm.4.2]{PossamaiTanZhou} for the specific family of measures
$\Pred_{[\underline{a},\overline{a}]}$, the details being deferred to the appendix. 
We employ the classical sign convention for standard BSDE generators according to which the generator of the BSDE (\ref{eq:BSDEStandard}) below is $-\widehat{F}$  
(i.e.\ with a minus sign). The convention for associated 2BSDE generator however remains as already introduced, e.g.\ the 2BSDE (\ref{eq:2BSDE}) 
has the generator $F$. 
\begin{proposition}\label{pro:ExistenceUniquenessSol2BSDEs}
 If $X\in \mathbb{L}^2(\filt_+)$ and Assumption~\ref{asp:Assumption1}, then the 2BSDE (\ref{eq:2BSDE}) has a
 \\
1.\  unique solution $(Y,Z,(K^P)_{P\in\Pred_{[\underline{a},\overline{a}]}}) \in  \mathbb{D}^2\big(\filt^{\Pred_{[\underline{a},\overline{a}]}}\big)\times \mathbb{H}^2\big(\filt^{\Pred_{[\underline{a},\overline{a}]}}\big)\times \mathbb{I}^2\big(\big(\filt^P\big)_{P\in\Pred_{[\underline{a},\overline{a}]}}\big)$
\\
2.\ and for any $P\in\Pred_{[\underline{a},\overline{a}]}$, the $Y$-part of the solution 
 has the representation 
 \begin{equation}\label{eq:RepSol2BSDEs}
  Y_s = \esssup^{\qquad\quad P}_{P'\in\Pred_{[\underline{a},\overline{a}]}(s,P,\filt_+)}\mathcal{Y}^{P'}_s(t,Y_t),\ s\le t\le T,\ P\text{-a.s.},
 \end{equation}
where $(\mathcal{Y}^{P}(\tau,H),\mathcal{Z}^{P}(\tau,H))$ denotes the unique solution to the standard BSDE 
\begin{equation}\label{eq:BSDEStandard}
  \mathcal{Y}^{P}_t= H - \int_t^\tau\widehat{F}_s(\mathcal{Y}^{P}_s,\widehat{a}^{1/2}_s\mathcal{Z}^{P}_s)ds - \phantom{}^{\phantom{}^{(P)}}\hspace{-0.15cm}\hspace{-0.1cm}\int_t^\tau(\mathcal{Z}^{P}_s)^{\text{tr}}dB_s, \ t\le \tau,\ P\text{-a.s.},
\end{equation}
with parameters $(-\widehat{F},H)$, for an $\filt^P$-stopping time $\tau$ and $H\in L^2(\F^P_{\tau},P)$. 
\end{proposition}
\begin{remark}
Note in comparison to \citet{SonerTouziZhang-Wellposedness} that uniform continuity (in $(\omega)$) and convexity (in $a$) of the generator function $F$ are not required 
for the more general wellposedness results of \citet{PossamaiTanZhou} summarized here in Proposition \ref{pro:ExistenceUniquenessSol2BSDEs}.
The latter is what we need in Section \ref{sec:GoodDealValandHedgingwithVolUncer}  for applications to valuation and hedging under combined drift and volatility uncertainty, as the results of \citet{SonerTouziZhang-Wellposedness} 
may not be applicable in that situation; cf.\ Part 1 of Remark \ref{rem:DynProgPrinciple} for a detailed justification. Note that the generalized theory also works for terminal conditions 
which are merely Borel measurable and do not need to be in the closure of uniformly continuous functions as required in \citet{SonerTouziZhang-Wellposedness}. 
\end{remark}

\section{Financial market model and good-deal constraints}\label{sec:NGDRestrictionandFinMarketVolUncer}
We apply the 2BSDE theory of Section \ref{sec:2BSDEs} to good-deal valuation and hedging of contingent claims in incomplete financial markets under drift and volatility uncertainty. 
In comparison to standard BSDEs which are used in \citet{BechererKentiaTonleu} in the presence of solely drift uncertainty, 2BSDEs are an appropriate tool for describing worst-case 
valuations in the presence of volatility uncertainty. 
As in \citet{CochraneRequejo,BjorkSlinko}, we consider good-deal constraints imposed as bounds on the Sharpe ratios
\cite[equivalently bounds on the optimal growth rates as in][]{Becherer-Good-Deals}
in the financial market extended by further (derivative) asset price processes. 
But first we specify the model for the market with uncertainty about the  volatility and the market price of risk.
\subsection{Financial market with combined uncertainty about drift and volatility}\label{subsec:FinanMarkModelVolUncer}
The financial market consists of $d$ tradeable stocks ($d\le n$)
with discounted price processes $(S^i)_{i=1}^d=S$ modeled by 
\begin{equation*}
 dS_t = \mathrm{diag}(S_t)(b_tdt +\sigma_tdB_t),\ t\in[0,T] ,\ \Pred_{[\underline{a},\overline{a}]}\text{-q.s.},\quad S_0\in (0,\infty)^d,
\end{equation*}
where $b$ (resp.\ $\sigma$) is a $\R^d$-valued (resp.\ $\R^{d\times n}$-valued) $\filt$-predictable uniformly bounded process, with $\sigma$ being such that 
\begin{equation}\label{eq:AggregCond}
\text{the family } \Big\lbrace\phantom{}^{\phantom{}^{(P)}}\hspace{-0.15cm}\hspace{-0.1cm}\int_0^\cdot\sigma_sdB_s,\, P\in\Pred_{[\underline{a},\overline{a}]}\Big\rbrace \text{ aggregates into a single process } \int_0^\cdot\sigma_sdB_s.
\end{equation}
We assume in addition that $\sigma\sigma^{\text{tr}}$ is uniformly elliptic in the sense that
\begin{equation}\label{eq:SigmaUnifEll}
 \text{there exists } \Upsilon,\Lambda\in(0,\infty)\ \text{such that}\quad \Upsilon\,\text{I}_{d\times d}\le \sigma\sigma^{\text{tr}}\le \Lambda\,\text{I}_{d\times d},\quad \Pred_{[\underline{a},\overline{a}]}\otimes dt\text{-q.e.},
\end{equation}
where $\text{I}_{d\times d}$ denotes the $d\times d$ identity matrix. In particular $\sigma\widehat{a}^{1/2}$ is $\Pred_{[\underline{a},\overline{a}]}\otimes dt$-q.e.\ of maximal rank $d\le n$, since $\sigma\widehat{a}\sigma^{\text{tr}}$
is uniformly elliptic and bounded (by (\ref{eq:DefPH}),(\ref{eq:SigmaUnifEll})). 
\begin{remark}\label{rem:AboutModel}
  For economic interpretation, one should clearly have aggregation of $S$, which should be quasi-surely defined as a single process. The latter is ensured here by the aggregation condition 
  (\ref{eq:AggregCond}) which might seem restrictive at first sight, but is ensured for instance if $\sigma$ is c\`adl\`ag in which case $\int_0^\cdot\sigma_sdB_s$ 
  can even be constructed pathwise as in \citet{Karandikar}.
\end{remark}
The market model captures uncertainty about the volatility of the stock prices $S$ which is $\sigma\widehat{a}^{1/2}$ under each measure $P\in\Pred_{[\underline{a},\overline{a}]}$.
Since $dB_t=\widehat{a}^{1/2}_tdW^P_t\ P$-a.s.\ for a $P$-Brownian motion $W^P$, then the dynamics of $(S_t)_{t\in[0,T]}$ under $P\in \Pred_{[\underline{a},\overline{a}]}$ is 
\begin{align*}
 dS_t = \mathrm{diag}(S_t)\sigma_t\widehat{a}^{1/2}_t(\widehat{\xi}_tdt +dW^P_t),\ 
\text{with } 
\widehat{\xi}:=\widehat{a}^{1/2}\sigma^{\text{tr}}(\sigma\widehat{a}\sigma^{\text{tr}})^{-1}b 
\end{align*}
denoting the market price of risk in each reference model $P\in\Pred_{[\underline{a},\overline{a}]}$. Note that $\widehat{\xi}$ is $\R^n$-valued, $\filt$-predictable and uniformly bounded by 
a constant depending only on $\underline{a},\overline{a},\Lambda,\Upsilon$, and the uniform bound on $b$. 
The financial market described is thus typically incomplete for any reference measure $P\in\Pred_{[\underline{a},\overline{a}]}$ for the volatility $\sigma\widehat{a}^{1/2}$ if $d<n$.
In practice, the bounds $\underline{a},\bar{a}$ and the uniform bounds on $\sigma\sigma^{\text{tr}}$ can be viewed as describing some confidence region for future volatility values, 
which might be set e.g.\ according to expert opinion about the range of historical or future (implied) volatility scenarios. 

To incorporate also uncertainty about the drift, we admit for market prices of risk $\widehat{\xi}^\theta$
being from a radial set of which $\widehat{\xi}$ is the center, that is we consider 
\begin{equation}\label{eq:SetCandMPR}
\Big\lbrace \widehat{\xi}^{\,\theta}:=\widehat{\xi}+\widehat{\Pi}(\theta)\ \Big \lvert\ \theta\ \filt\text{-predictable with }\lvert\theta_t(\omega)\rvert\le \delta_t(\omega),\ \text{} (t,\omega)\in[0,T]\times\Omega\Big\rbrace,
\end{equation}
where 
$\widehat{\Pi}_{(t,\omega)}(z):=(\sigma_t(\omega)\widehat{a}_t^{1/2}(\omega))^{\text{tr}}(\sigma_t(\omega)\widehat{a}_t(\omega)\sigma_t^{\text{tr}}(\omega))^{-1}(\sigma_t(\omega)\widehat{a}_t^{1/2}(\omega))z$ denotes the orthogonal 
projection of $z\in\R^n$ onto $\mathrm{Im}\,(\widehat{a}^{1/2}_t(\omega)\sigma^{\text{tr}}_t(\omega)),\ t\in[0,T]$, and $\delta$ is a fixed non-negative bounded $\filt$-predictable 
process. 
The set (\ref{eq:SetCandMPR}) of 
market prices of risk corresponds to an ellipsoidal confidence region 
of drift uncertainty on risky asset prices $S$, such that ambiguous drifts could attain values in ellipsoids 
\begin{equation*}
\Big\lbrace x\in\R^d:\ \big(x-b_t(\omega)\big)^{\text{tr}}\big(\sigma_t(\omega)\widehat{a}_t(\omega)\sigma_t^{\text{tr}}(\omega)\big)^{-1}\big(x-b_t(\omega)\big)\le \delta^2_t(\omega)\Big\rbrace,
\end{equation*}
at $(t,\omega)\in[0,T]\times\Omega$.
Ellipsoidal specifications of uncertainty are common in the literature, and appear naturally in the context of uncertainty about the drifts of tradeable asset prices  in multivariate Gaussian settings, 
cf.\ e.g.\ \citet[][]{Garlappi,BiaginiPinar}. 

We denote by $\Theta:[0,T]\times\Omega\rightsquigarrow \R^n$ the correspondence (set-valued mapping) 
\begin{equation*}
\Theta_t(\omega):=\Big\lbrace x\in\R^n:\ \lvert x\rvert\le \delta_t(\omega)\Big\rbrace,\quad \text{for } (t,\omega)\in[0,T]\times\Omega,
\end{equation*}
with radial values. The notation ``$\rightsquigarrow$'' emphasizes that $\Theta$ is set-valued. 
$\filt$-predictability of $\delta$ implies that $\Theta$ is $\filt$-predictable 
in the sense of \citet{Rockafellar}, i.e.\ for each closed set $F\subset \R^n$ the set $\Theta^{-1}(F):=\{(t,\omega)\in[0,T]\times\Omega: \Theta_t(\omega)\cap F\neq \emptyset \}$ is $\filt$-predictable.
Hence by measurable selection arguments \citep[e.g.][Cor.1.Q]{Rockafellar}, $\Theta$ admits $\filt$-predictable selections, i.e.\ $\filt$-predictable functions $\theta$ satisfying $\theta_t(\omega)\in \Theta_t(\omega)$ for all $[0,T]\times\Omega.$
For arbitrary correspondence $\Gamma:[0,T]\times\Omega\rightsquigarrow \R^n$, we will shortly write $\lambda\in \Gamma$ to mean that the function $\lambda$ is a $\filt$-predictable selection of $\Gamma$.
Moreover we will say that $\lambda$ is selection of $\Gamma$ (not necessarily measurable) if $\lambda_t(\omega)\in\Gamma_t(\omega)$ for all $(t,\omega)\in[0,T]\times\Omega$.

Combined Knightian uncertainty (i.e.\ ambiguity) about drift and volatility scenarios is then captured by a (typically non-dominated) set 
\begin{equation*}
\mathcal{R} := \Big\lbrace Q: Q\sim P,\ {dQ} = \phantom{}^{(P)}\hspace{-0.05cm}\mathcal{E}\big( \theta\cdot W^P\big){dP}\text{ for some }\theta\in\Theta\text{ and } P\in\Pred_{[\underline{a},\overline{a}]}\Big\rbrace 
\end{equation*}
of candidate reference probability measures (priors), where $\phantom{}^{(P)}\hspace{-0.05cm}\mathcal{E}(M):= \exp\big(M -M_0- \frac{1}{2}\langle M\rangle^P\big)$ denotes the stochastic exponential of a local martingale $M$ under $P$.
For any $Q\in\mathcal{R}$, there are $P^Q\in \Pred_{[\underline{a},\overline{a}]}$, $\theta^Q\in\Theta$ such that the canonical process $B$ is a $Q$-semimartingale with 
decomposition $B= \int_0^\cdot\widehat{a}^{1/2}_s\theta^Q_sds+\int_0^\cdot\widehat{a}^{1/2}_sdW^Q_s$,
where $W^Q = W^P-\int_0^\cdot\theta^Q_sds$ is a $Q$-Brownian motion. We will simply denote by $Q^{P,\theta}$ a reference measure $Q\in\mathcal{R}$ associated to $P\in\Pred_{[\underline{a},\overline{a}]}$
and $\theta\in\Theta$, and we take note that the specification range for $P$ and $\theta$, respectively, 
 account for uncertainty about volatilities and drifts, respectively.
$S$ evolves under $Q^{P,\theta}\in\mathcal{R}$,  $P$-a.s.,  as
\begin{equation*}
	 dS_t = \mathrm{diag}(S_t)\sigma_t\widehat{a}^{1/2}_t(\widehat{\xi}^{\,\theta}_tdt +dW^{P,\theta}_t),
\end{equation*}
where $W^{P,\theta}:=W^P-\int_0^\cdot\theta_sds$ is a $Q^{P,\theta}$-Brownian motion. Hence $\widehat{\xi}^{\,\theta}$ is the market price of risk in the model $Q^{P,\theta}$, with volatility $\sigma\widehat{a}^{1/2}$
and $\widehat{a}$ satisfying $\underline{a}\le \widehat{a}\le \overline{a},\ P\otimes dt\text{-a.e.}$.
Let $\mathcal{M}^e(Q) := \mathcal{M}^e(S,Q)$ denote the set of equivalent local martingale measures for $S$ in a model $Q$. 
Then standard arguments \citep[analogously to ][Prop.4.1]{BechererKentiaTonleu} easily lead to the following 
\begin{lemma}\label{lem:EMMP}
 For any $P\in\Pred_{[\underline{a},\overline{a}]}$ and $\theta\in \Theta$, the set $\mathcal{M}^e(Q^{P,\theta})$ is equal to 
 \begin{align*}
 	 &
 	 \Big\lbrace Q\sim Q^{P,\theta}\ \Big\lvert\Big.\ dQ = \phantom{}^{(P)}\hspace{-0.05cm}\mathcal{E}(\lambda\cdot W^{P,\theta})\,dQ^{P,\theta},\ \text{}\lambda = -\widehat{\xi}^{\,\theta}+\eta,\ \eta\in\mathrm{Ker}\,(\sigma\widehat{a}^{1/2})\Big\rbrace \\
 &= \Big\lbrace Q\sim P \ \Big\lvert\Big.\  dQ = \phantom{}^{(P)}\hspace{-0.05cm}\mathcal{E}(\lambda\cdot W^{P})\,dP,\ \text{} \lambda = -\widehat{\xi}+\eta,\ \eta\in\mathrm{Ker}\,(\sigma\widehat{a}^{1/2})\Big\rbrace 
				    =\mathcal{M}^e(P).
 \end{align*}
\end{lemma}
\begin{remark}\phantomsection{}\label{rem:AboutZeroDrift}
Note by Lemma \ref{lem:EMMP} that for each $P\in\Pred_{[\underline{a},\overline{a}]}$ and $\theta\in\Theta$, the minimal martingale measure $\widehat{Q}^{P,\theta}$ given by 
$d\widehat{Q}^{P,\theta} = \phantom{}^{(P)}\hspace{-0.05cm}\mathcal{E}(-\widehat{\xi}^{\,\theta}\cdot W^{P,\theta})\,dQ^{P,\theta}$ is in $\mathcal{M}^e(P)$ since $\widehat{\xi}$ and $\delta$ are uniformly bounded. 
This implies that 
$\mathcal{M}^e(Q)\neq\emptyset$ for any $Q\in\mathcal{R}$. Thus, the market satisfies the no-free lunch with vanishing risk 
 condition \citep[][]{DelbaenSchachermayer} under any (uncertain) prior $Q\in\mathcal{R}$. We will interpret this as a notion for no-arbitrage under drift and volatility 
 uncertainty \citep[as in e.g.][]{BiaginiBouchardKardarasNutz}.
\end{remark}
We parametrize trading strategies $\varphi = (\varphi^i)_{i=1}^d$ in terms of amount $\varphi^i$ of wealth invested in the stock
with price process $S^i$, such that $\varphi$ is a $\filt^{\Pred_{[\underline{a},\overline{a}]}}$-predictable process satisfying suitable integrability properties that will be made precise. 
In this respect, the wealth process $V^\varphi$ associated to a 
trading strategy $\varphi$ with initial capital $V_0$ (so that $(V_0,\varphi)$ quasi-surely satisfies the self-financing requirement) would have dynamics, on $[0,T]$, 
\begin{align*}
 V^\varphi_t &=V_0 +\int_0^t\varphi^{\text{tr}}_s(b_sds+\sigma_sdB_s) = V_0 +\int_0^t\varphi^{\text{tr}}_s\sigma_s(\widehat{a}^{1/2}_s\widehat{\xi}_sds+dB_s)   ,\ \Pred_{[\underline{a},\overline{a}]}\text{-q.s.}.
 \end{align*}
Re-parameterizing trading strategies in terms of integrands $\phi:=\sigma^{\text{tr}}\varphi \in\mathrm{Im}\,\sigma^{\text{tr}}$ w.r.t.\  $B+\int_0^\cdot\widehat{a}^{1/2}_s\widehat{\xi}_sds$ yields as dynamics for the wealth process $V^\phi:=V^\varphi$ 

\begin{equation*}
 V^\phi_t = V_0 + \int_0^t\phi^{\text{tr}}_s(\widehat{a}^{1/2}_s\widehat{\xi}_sds+dB_s) =V_0 + \phantom{}^{\phantom{}^{(P)}}\hspace{-0.15cm}\int_0^t\phi^{\text{tr}}_s\widehat{a}_s^{\frac{1}{2}}(\widehat{\xi}_sds+dW^P_s),\ t\in[0,T] ,\ P\text{-a.s.}
\end{equation*}
for any $P\in\Pred_{[\underline{a},\overline{a}]}$.
The set $\Phi(P)$ of permitted strategies in the model $P$ (shortly $P$-permitted), $P\in\Pred_{[\underline{a},\overline{a}]}$, is defined as 
\begin{equation*}
 \Phi(P):=\big\lbrace \phi \in \mathbb{H}^2(\filt^P,P): \phi_t(\omega)\in\mathrm{Im}\ \sigma^{\text{tr}}_t(\omega),\ \text{for all } (t,\omega)\in[0,T]\times\Omega\big\rbrace.
\end{equation*}
Under uncertainty and i.p.\ for non-dominated priors $P\in\Pred_{[\underline{a},\overline{a}]}$, we want the wealth process for a trading strategy to be defined (q.e.) as single process, and not to vary with the prior. 
This requires as additional  condition on  strategies $\phi$ that the family $\big\lbrace\phantom{}^{(P)}\hspace{-0.1cm}\int_0^\cdot\phi^{\text{tr}}_sdB_s,\ P\in \Pred_{[\underline{a},\overline{a}]}\big\rbrace$ of 
``profit\&loss''-processes aggregates  (cf.\ p.\pageref{aggrnotion}) into a single process, denoted $\int_0^\cdot\phi^{\text{tr}}_sdB_s$. To this end, we make the
\begin{definition}\label{eq:DefPhi}
	The set $\Phi$ of permitted trading strategies under drift and volatility uncertainty consists of all processes $\phi$ in $\mathbb{H}^2\big(\filt^{\Pred_{[\underline{a},\overline{a}]}}\big)$
	with $\phi_t(\omega)\in\mathrm{Im}\,\sigma^{\text{tr}}_t(\omega)$ for all $(t,\omega)\in[0,T]\times\Omega$ and such that 
the family $\big\lbrace \phantom{}^{(P)}\hspace{-0.1cm}\int_0^\cdot\phi^{\text{tr}}_sdB_s,\, P\in\Pred_{[\underline{a},\overline{a}]}\big\rbrace$ aggregates into a 
	single process $\int_0^\cdot\phi^{\text{tr}}_sdB_s$. 
\end{definition}
The trading strategies in $\Phi$ will be termed as $\Pred_{[\underline{a},\overline{a}]}$-permitted (in short: permitted) strategies. 
By definition, one has $\Phi\subseteq \Phi(P)$ for all $P\in\Pred_{[\underline{a},\overline{a}]}$. 
Hence $V^\phi$ is a martingale under some measure equivalent to $P$ for any $\phi\in\Phi$ and $P\in\Pred_{[\underline{a},\overline{a}]}$. This excludes arbitrage strategies 
from $\Phi$ for any scenario $\sigma \widehat{a}^{1/2}$ of the volatility. 
\begin{remark}\label{rem:OnAggPhi}
	Note that by a result of \citet{Karandikar} the aggregation property of Definition \ref{eq:DefPhi} would be satisfied for $\phi$ being c\`adl\`ag and $\filt$-adapted (like in  Remark~\ref{rem:AboutModel} 
 for the volatility $\sigma$).
	Alternatively, aggregation would automatically hold under additional (non-standard) axioms for set-theory as in \citet{Nutz}; cf.\  Remark \ref{rem:OnAggregationOfK}.
\end{remark}
\subsection{No-good-deal restriction and implications to dynamic valuation and hedging}\label{subsect:NGDRestriction}
In the absence of uncertainty, we consider a (classical) no-good-deal restriction defined as a bound on the instantaneous Sharpe ratios, for any extension of the market by additional 
derivatives' prices computed by the no-good-deal pricing measures, following  \citet{BjorkSlinko,CochraneRequejo}.  
In a setup without jump, such a constraint is equivalent to a bound on the optimal expected growth rate of returns in any market extension, as in e.g.\ \citet{Becherer-Good-Deals}. 
It is known that such is ensured by imposing a bound on the norm of the Girsanov kernels from the risk-neutral pricing measures. 
As slightly more detailed elaboration on the features of this valuation approach is given in Remark \ref{refsforngdrestrictions}. 
Under solely drift uncertainty, robust good-deal hedging has been studied by \citet{BechererKentiaTonleu}.
To extend the theory by including moreover ambiguity on volatility, we impose as no-good-deal restriction under combined  drift and volatility 
uncertainty now the same Sharpe ratio bound but require it under each model $Q\in\mathcal{R}$ separately. By doing so, we obtain for each $Q\in\mathcal{R}$ a set of 
no-good-deal measures $\mathcal{Q}^{\text{ngd}}(Q)\subseteq \mathcal{M}^e(Q)$. Then following a worst-case approach to good-deal valuation under uncertainty 
we define the robust good-deal bound under drift and volatility uncertainty as the supremum over all no-good-deal prices for  all reference priors $Q\in\mathcal{R}$.

To explain the idea for the good-deal approach, we will next do so at first without addressing model uncertainty. 
To this end, let us assume just for the remaining of this subsection that we neither have ambiguity about drifts nor about volatility and let us consider,  just as a starting point, one arbitrary probability 
reference measure   $Q^{P,\theta}$ from $ \mathcal{R}$ to be the objective real world measure for the standard (non-robust) good deal problem of valuation and hedging. This means that we take one specification $P$ and $\theta$ 
for volatility and drift to be given. Having explained the problem without ambiguity, this will set the stage to explain the robust problem with model uncertainty thereafter in the main Section~\ref{sec:GoodDealValandHedgingwithVolUncer}.

\subsubsection{The no-good-deal problem under a given model $Q^{P,\theta}$  without uncertainty}

Let $h$ be a fixed non-negative bounded $\filt$-predictable process. 
For a given probability measure $Q^{P,\theta}$ in $ \mathcal{R}$ (with $P\in\Pred_{[\underline{a},\overline{a}]} $ and $\theta\in\Theta$), we consider the set $\mathcal{Q}^{\text{ngd}}(Q^{P,\theta})$ of no-good-deal measures in the model 
$Q^{P,\theta}$ as the subset of $\mathcal{M}^e(Q^{P,\theta})$ consisting of all equivalent local martingale measures 
$Q$, 
whose Girsanov kernels $\lambda$ w.r.t.\  the $Q^{P,\theta}$-Brownian motion $W^{P,\theta}$ are bounded by $h$ in the sense that $\lvert\lambda_t(\omega)\rvert \le h_t(\omega)$ for all $(t,\omega)\in[0,T]\times\Omega$. 
More precisely using Lemma \ref{lem:EMMP}, we define 
\begin{equation}\label{eq:NGDMeasures}
	    \begin{split}
	      \mathcal{Q}^{\text{ngd}}(Q^{P,\theta}):=\Big\lbrace Q &\sim  Q^{P,\theta}\ \Big\lvert\  \frac{dQ}{dQ^{P,\theta}}=\phantom{}^{(P)}\hspace{-0.05cm}\mathcal{E}\big(\lambda\cdot W^{P,\theta}\big),\text{ for a } \filt\text{-predictable } 
\Big.\\
					    &\Big.\lambda = -\widehat{\xi}^{\,\theta}+\eta \text{ with }\lvert \lambda\rvert \le h \ \text{ for }\eta\in \mathrm{Ker}\,\big(\sigma\widehat{a}^{\frac 1 2}\big) \Big\rbrace. 
	      \end{split}
\end{equation}
In the sequel, any selection $\theta$ (not necessarily measurable)  of $\Theta$ is taken to satisfy 
\begin{equation}\label{eq:AspOnxitheta}
	\big\lvert\widehat{\xi}^{\,\theta}_t(\omega)\big\rvert\le h_t(\omega),\text{ for all }(t,\omega)\in[0,T]\times\Omega,\ \theta\in\Theta.
\end{equation}
Clearly, (\ref{eq:AspOnxitheta}) implies that the minimal martingale measure $\widehat{Q}^{P,\theta}$ is in $\mathcal{Q}^{\text{ngd}}(Q^{P,\theta})\neq \emptyset$ for any $P\in\Pred_{[\underline{a},\overline{a}]},\ \theta\in\Theta$. 
Beyond this (\ref{eq:AspOnxitheta}) will also be used to prove wellposedness of a 2BSDE whose solution will describe the robust good-deal valuation bound and hedging strategy (see Theorems \ref{thm:2BSDECHaractforPiU} and \ref{thm:GDHedgingTheorem}, respectively).
As in \citep[][Lem.2.1, part b)]{BechererKentiaTonleu}, one can show for any $P\in\Pred_{[\underline{a},\overline{a}]}$ and $\theta\in\Theta$ that the set $\mathcal{Q}^{\text{ngd}}(Q^{P,\theta})$ is convex and 
multiplicatively stable (in short m-stable). The latter property, also referred to as rectangularity in the economic literature \citet{ChenEpstein}, is usually required for time-consistency 
of essential suprema of conditional expectations over priors; see \citet{Delbaen} for a general study.

For fixed $P\in\Pred_{[\underline{a},\overline{a}]}$ and $\theta\in\Theta$, the upper good-deal bound for a claim $X\in L^2(\F^P_T, P)$ in the model $Q^{P,\theta}$, for given drift and volatility specification, is defined  as
\begin{equation}\label{eq:GDBfixedPrior}
\pi^{u,P,\theta}_t(X):=\esssup^{\qquad\quad P}_{Q\in\mathcal{Q}^{\text{ngd}}(Q^{P,\theta})}E^Q_t[X],\ t\in[0,T],\ P\text{-a.s.}.
\end{equation}

\begin{remark}\label{refsforngdrestrictions}
According to classical good-deal theory in continuous time with continuous price dynamics, 
the set $\mathcal{Q}^{\text{ngd}}(Q^{P,\theta})$ of risk-neutral measures $Q$ is specified in such a way, that in 
any extension $(S,S')$ of the financial market by additional derivatives' price processes $S'$, which are (local) $Q$-martingales,
the extended market is not just arbitrage-free but moreover does not permit opportunities for dynamic trading with overly attractive reward to risk ratios. More precisely, the choice of $\mathcal{Q}^{\text{ngd}}(Q^{P,\theta})$ is such that 
the extended market does not offer trading opportunities that yield instantaneous Sharpe ratios which exceed the bound $h$ \citep[see][Section~3 and Appendix~A, for details]{BjorkSlinko}; 
Equivalently, one can show that the extended market
does not admit for portfolio strategies that offer (conditional)  expected growth rates larger than $h$ \citep[see][Section~3]{Becherer-Good-Deals}, and the respective good-deal restrictions are sharp.
In this sense, the dynamic upper $\pi^{u,P,\theta}_t(X)$ and lower $-\pi^{u,P,\theta}_t(-X)$ good-deal bounds determine at any time $t$ a sub-interval of the arbitrage-free prices, that is determined by those valuation measures $Q$ which restrict the opportunities for overly good deals.
\end{remark}

The hedging objective for the seller of a claim $X\in L^2(P)$ who believes in the model $Q^{P,\theta}$ is to find a trading strategy $\bar{\phi}^{P,\theta}\in \Phi(P)$ that minimizes the residual risk 
under a risk measure $\rho^{P,\theta}$ from holding $X$ and trading dynamically in the market. As the seller charges a premium $\pi^{u,P,\theta}_\cdot(X)$ for $X$, she would like
this premium to be the minimal capital that dynamically makes her position acceptable, so that $\pi^{u,P,\theta}$ becomes the market consistent risk measure 
corresponding to $\rho^{P,\theta}$, in the spirit of \citet{BarrieuElKaroui}. Following \citet{Becherer-Good-Deals}, define $\rho^{P,\theta}$ as a dynamic coherent risk measure by 
\begin{equation}\label{eq:DefRhoPtheta}
\rho^{P,\theta}_t(X):=\esssup^{\qquad\quad P}_{Q\in\mathcal{P}^{\text{ngd}}(Q^{P,\theta})}E^Q_t[X],\ t\in[0,T],\ P\text{-a.s.}, \quad\text{  with}
\end{equation}
\begin{equation*}
	    \begin{split}
	      \mathcal{P}^{\text{ngd}}(Q^{P,\theta}):=\Big\lbrace Q\sim P\ \Big\lvert\ &\frac{dQ}{dQ^{P,\theta}}=\phantom{}^{(P)}\hspace{-0.05cm}\mathcal{E}\big(\lambda\cdot W^{P,\theta}\big),\ \lambda\ \filt\text{-prog., }\lvert \lambda\rvert \le h\ P\otimes dt\text{-a.e.}\Big\rbrace.
	      \end{split}
\end{equation*}
Hence $\mathcal{P}^{\text{ngd}}(Q^{P,\theta})$ is the set of a-priori valuation measures equivalent to $Q^{P,\theta}$, which satisfy the no-good-deal restriction under $Q^{P,\theta}$ 
but might not be local martingale measures for the stock price process $S$ (yet they are w.r.t.\ the market with only the riskless asset $S^0\equiv1$). 
This implies $\mathcal{Q}^{\text{ngd}}(Q^{P,\theta}) = \mathcal{P}^{\text{ngd}}(Q^{P,\theta})\cap \mathcal{M}^{e}(Q^{P,\theta})$,  
and thus $\rho^{P,\theta}_t(X)\ge\pi^{u,P,\theta}_t(X),\ t\in[0,T],\ P$-a.s.\ for all $X\in L^2(P)$. Moreover as $\mathcal{Q}^{\text{ngd}}(Q^{P,\theta})$ is m-stable and convex, the set $\mathcal{P}^{\text{ngd}}(Q^{P,\theta})$ also is, 
for any $P\in\Pred_{[\underline{a},\overline{a}]}$ and $ \theta\in\Theta$.
Therefore by \citep[][Lem.2.1]{BechererKentiaTonleu} the dynamic (coherent) risk measure $\rho^{P,\theta}:L^2(\F^P_T,P)\to L^2(\F^P_t,P)$ is time-consistent. 
For a fixed reference measure $Q^{P,\theta}$,  the hedging problem of the seller for $X\in L^2(\F^P_T,P)$ is then to find $\bar{\phi}^{P,\theta}\in\Phi(P)$ such that 
$P$-almost surely for all $t\in[0,T]$ holds
\begin{equation}\label{eq:GD-hedging-Problem-P-Theta}
\begin{split}
	\pi^{u,P,\theta}_t(X)&= \essinf^{\qquad\quad P}_{\phi\in\Phi(P)}\rho^{P,\theta}_t\Big( X-\phantom{}^{\phantom{}^{(P)}}\hspace{-0.15cm}\int_t^T\phi^{\text{tr}}_s\,\widehat{a}_s^{\frac{1}{2}}\big(\widehat{\xi}_sds+dW^{P}_s\big)\Big)\\
		  &=\rho^{P,\theta}_t\Big( X-\phantom{}^{\phantom{}^{(P)}}\hspace{-0.15cm}\int_t^T(\bar{\phi}^{P,\theta}_s)^{\text{tr}}\,\widehat{a}_s^{\frac{1}{2}}\big(\widehat{\xi}_sds+dW^{P}_s\big)\Big).
\end{split}
\end{equation}

One can show (cf.\  Cor.5.6 in \citet{Becherer-Good-Deals}, or Prop.4  in \citet{BechererKentiaTonleu}) that for each $P$ in $\Pred_{[\underline{a},\overline{a}]}$ and $\theta$ in $\Theta$, the tracking (or hedging) error 
\begin{equation*}
R^{\bar{\phi}^{P,\theta}}(X):= \pi^{u,P,\theta}_\cdot(X) - \pi^{u,P,\theta}_0(X)-\phantom{}^{\phantom{}^{(P)}}\hspace{-0.15cm}\int_0^\cdot(\bar{\phi}^{P,\theta}_s)^{\text{tr}}\,\widehat{a}_s^{\frac{1}{2}}\big(\widehat{\xi}_sds+dW^{P}_s\big)
\end{equation*}
of the good-deal hedging strategy $\bar{\phi}^{P,\theta}$ is a supermartingale under any  measure $Q\in \mathcal{P}^{\text{ngd}}(Q^{P,\theta})$. This supermartingale property of tracking errors
could be viewed as a robustness property of the hedging strategy with respect to the family $\mathcal{P}^{\text{ngd}}(Q^{P,\theta})$ of 'valuation' probability measures as generalized scenarios \citep[cf.][]{ArtznerDelbaenEberHeath}.
\subsubsection{Standard BSDEs for valuation and hedging under $Q^{P,\theta}$  without uncertainty}
The solution to the valuation and hedging problem described by (\ref{eq:GDBfixedPrior}) and (\ref{eq:GD-hedging-Problem-P-Theta}) can be obtained in terms of standard BSDEs
under $P\in\Pred_{[\underline{a},\overline{a}]}$. In order to be more precise, let us introduce some notations that will also be used throughout the sequel.
For $a\in\mathbb{S}_n^{>0}$, we denote by $\Pi^a_{(t,\omega)}(\cdot)$ and $\Pi^{a,\bot}_{(t,\omega)}(\cdot)$ respectively the orthogonal projections onto the 
subspaces $\mathrm{Im}\,\big(\sigma_t(\omega)a^{\frac{1}{2}}\big)^{\text{tr}}$ and $\mathrm{Ker}\,\big(\sigma_t(\omega)a^{\frac{1}{2}}\big)$ of $\R^n,$ $(t,\omega)\in[0,T]\times\Omega $.
Explicitly, for each $a\in\mathbb{S}_n^{>0}$ and $t\in[0,T]$  (omitting $\omega$-symbols for simplicity), the projections for $z\in\R^n$ are given by
 \begin{equation*}
  \Pi^a_t(z) = (\sigma_ta^{1/2})^{\text{tr}}(\sigma_ta\sigma_t^{\text{tr}})^{-1}(\sigma_ta^{1/2})z
  \quad\text{and}\quad \Pi^{a,\bot}_t(z) = z-\Pi^a_t(z).
 \end{equation*}
In particular we define $\widehat{\Pi}_{(t,\omega)}(\cdot) := \Pi^{\widehat{a}_t(\omega)}_{(t,\omega)}(\cdot)$ and $\widehat{\Pi}^\bot_{(t,\omega)}(\cdot):=\Pi^{\widehat{a}_t(\omega),\bot}_{(t,\omega)}(\cdot)$. 
For $P\in\Pred_{[\underline{a},\overline{a}]}$ and $\theta\in\Theta$, let us consider the standard BSDE under $P$,
\begin{equation}\label{eq:BSDEValandHedging-fixed-P-theta}
\mathcal{Y}_t= X - \int_t^T\widehat{F}_s^{\,\theta}(\mathcal{Y}_s,\widehat{a}^{1/2}_s\mathcal{Z}_s)ds - \phantom{}^{\phantom{}^{(P)}}\hspace{-0.15cm}\int_t^T\mathcal{Z}_s^{\text{tr}}dB_s, \ t\in[0,T],\ P\text{-a.s.},  	
\end{equation}
for data $X\in L^2(\F_T^P,P)$ and generator $-\widehat{F}_t^{\,\theta}(\omega,\cdot) = -F^\theta(t,\omega_{\cdot\wedge t},\cdot,\widehat{a}_t(\omega))$ with
\begin{equation}\label{eq:GeneratorFunctionFtheta}
 \begin{split}
 	F^\theta(t,\omega,z,a) := &-\Pi^{a,\bot}_{(t,\omega)}\big(\theta_t(\omega)\big)^\textrm{tr}\Pi^{a,\bot}_{(t,\omega)}(z)+\widehat{\xi}_t^{\,\textrm{tr}}(\omega)\Pi^a_{(t,\omega)}(z)\\
					    &- \Big(h^2_t(\omega)-\big\lvert \widehat{\xi}_t(\omega)+\Pi^a_{(t,\omega)}(\theta_t(\omega))\big\rvert^2\Big)^{\frac{1}{2}}\big\lvert \Pi^{a,\bot}_{(t,\omega)}(z)\big\rvert
 \end{split}
\end{equation}
for all $(t,\omega,z,a)\in[0,T]\times\Omega\times\R^n\times\mathbb{S}^{>0}_n$. By \citep[][Thm.5.4]{Becherer-Good-Deals}, the upper valuation bound $\pi^{u,P,\theta}_\cdot(X)$ and 
good-deal hedging strategy $\bar{\phi}^{P,\theta}$ in the model $Q^{P,\theta}$ are identified in terms of the solution $(\mathcal{Y}^{P,\theta},\mathcal{Z}^{P,\theta})\in\mathbb{D}^2(\filt^P,P)\times\mathbb{H}^2(\filt^P,P)$ to the standard BSDE (\ref{eq:BSDEValandHedging-fixed-P-theta}) 
as $\pi^{u,P,\theta}_\cdot(X) = \mathcal{Y}^{P,\theta}$ and
\begin{equation*}
	\widehat{a}^{\frac{1}{2}}_t\,\bar{\phi}^{P,\theta}_t = \widehat{\Pi}_t\big(\widehat{a}_t^{\frac{1}{2}}\mathcal{Z}^{P,\theta}_t\big)+ \frac{\big\lvert \Pi^\bot_t\big(\widehat{a}_t^{\frac{1}{2}}\mathcal{Z}^{P,\theta}_t\big)\big\rvert}{\sqrt{h^2_t-\big\lvert\widehat{\xi}_t+\widehat{\Pi}_t\big(\theta_t\big)\big\rvert^2}}\Big(\widehat{\xi}_t+\widehat{\Pi}_t\big(\theta_t\big)\Big)\quad P\otimes dt\text{-a.e.}. 
\end{equation*}

As drift uncertainty can be incorporated in a setup with one dominating reference measure, the above standard BSDE description of good-deal bounds and hedging strategies 
has been extended by \citet{BechererKentiaTonleu} to the presence of uncertainty solely about the drift. Under uncertainty about the volatility, however, a dominating probability measure would fail to 
exist and hence standard BSDE techniques will not be applicable anymore. Instead, we will rely on 2BSDEs to provide in Section \ref{sec:GoodDealValandHedgingwithVolUncer}, under combined uncertainty about drifts 
and volatilities, an infinitesimal characterization of (robust) good-deal bounds and hedging strategies, after suitably defining the latter in such a typically non-dominated setup. 

\section{Good-deal hedging and valuation  under combined uncertainty}\label{sec:GoodDealValandHedgingwithVolUncer}
We describe good-deal bounds in the market model of Section \ref{subsec:FinanMarkModelVolUncer}  using 2BSDEs and study a corresponding notion of hedging that is robust w.r.t.\ combined drift and 
volatility uncertainty. We first define the good-deal valuation bounds from the no-good-deal restriction of Section \ref{subsect:NGDRestriction}, but taking into account 
the investor's aversion towards drift and volatility uncertainty. Furthermore, still as in Section \ref{subsect:NGDRestriction},
hedging strategies are defined as minimizers of some dynamic coherent risk measure $\rho$ under uncertainty, 
so that the good-deal bound arises as the market consistent risk measure for $\rho$, in the spirit of \citet{BarrieuElKaroui}, allowing for optimal risk sharing with the market under uncertainty.

For valuation, it seems natural to view uncertainty aversion as a penalization to the no-good-deal 
restriction in the sense that it implies good-deal bounds under uncertainty that are wider than in the absence of uncertainty. 
To formalize this idea, we rely on a classical worst-case  approach to uncertainty in  
the spirit of e.g.\ \citet{GilboaS89,HansenSargent}. 
The idea is that an agent who is averse towards ambiguity about the actual drifts and volatilities may opt, to be conservative, for a worst-case approach to valuation in order to compensate for losses that could occur 
due to the wrong choice of model parameters. Acting this way, she would sell (resp.\ buy) financial risks at the largest upper (resp.\ smallest lower) good-deal bounds over all plausible priors 
in her confidence set $\mathcal{R}$ capturing drift and volatility uncertainty. 
Given the technical difficulties that may arise from $\mathcal{R}$ being non-dominated, in particular for writing essential suprema, we define the  
worst-case upper good-deal bound $\pi^u_\cdot(X)$ for a financial risk $X\in L^2_{\Pred_{[\underline{a},\overline{a}]}}$ as the unique process $\pi^u_\cdot(X)\in\mathbb{D}^2\big(\filt^{\Pred_{[\underline{a},\overline{a}]}}\big)$ (if it exists) that satisfies 
\begin{equation}\label{eq:DefGDBounds}
 \pi^u_t(X)= \esssup^{\qquad\quad P}_{P'\in\Pred_{[\underline{a},\overline{a}]}(t,P,\filt_+)}\esssup^{\qquad\quad P}_{\theta\in\Theta}\pi^{u,P',\theta}_t(X),\quad t\le T,\ P\text{-a.s.},\ \text{}P\in\Pred_{[\underline{a},\overline{a}]},
\end{equation}
with $\pi^{u,P',\theta}_\cdot(X)$ defined in (\ref{eq:GDBfixedPrior}).
The lower bound $\pi^l_\cdot(X)=-\pi^u_\cdot(-X)$ is defined analogously, replacing essential suprema in (\ref{eq:DefGDBounds}) by essential infima, and
for this reason we focus only on studying the upper bound. For $X\in\mathbb{L}^2(\filt_+)$, we shown in Section \ref{subsec:ValuationBounds2BSDE} that (\ref{eq:DefGDBounds}) 
indeed defines a single process $\pi^u_\cdot(X)$ that can be identified as the $Y$-component of the solution of a specific 2BSDE. 

For robust good-deal hedging, we define the hedging strategy similarly as in Section \ref{subsect:NGDRestriction}, but taking into account model uncertainty. 
Indeed, under combined drift and volatility uncertainty, the hedging objective of the investor for a liability $X$ is to find a $\Pred_{[\underline{a},\overline{a}]}$-permitted 
trading strategy that dynamically minimizes the residual risk (under a suitable worst-case risk measure $\rho$) from holding $X$ and trading in the market. 
As a seller charging the premium $\pi^u_\cdot(X)$ for $X$, she would like the upper good-deal bound to be the minimal capital to make her  
position $\rho$-acceptable at all times so that $\pi^u_\cdot(\cdot)$ becomes the market consistent risk measure corresponding to $\rho_\cdot(\cdot)$.
The second objective of the investor being robustness (of hedges and valuations) w.r.t.\ ambiguity about both drifts and volatilities,
$\rho$ should be compatible with the no-good-deal restriction in the market and should also capture the investor's aversion towards uncertainty.
From the definition of $\pi^u_\cdot(X)$ (for $X\in \mathbb{L}^2(\filt_+)$) in (\ref{eq:DefGDBounds}) and the hedging problem (\ref{eq:GD-hedging-Problem-P-Theta}) 
in the absence of uncertainty, we define $\rho_\cdot(X)$ for $X\in L^2_{\Pred_{[\underline{a},\overline{a}]}}$ as the unique process in 
$\mathbb{D}^2\big(\filt^{\Pred_{[\underline{a},\overline{a}]}}\big)$ (if it exists) that satisfies, for $t\in [0,T]$, 
\begin{equation}\label{eq:DefRhoFirst}
 \rho_t(X)= \esssup^{\qquad\quad P}_{P'\in\Pred_{[\underline{a},\overline{a}]}(t,P,\filt_+)} \esssup^{\qquad\quad P}_{\theta\in\Theta}\rho^{P',\theta}_t(X),\ P\text{-a.s.\  for }P\in\Pred_{[\underline{a},\overline{a}]},
\end{equation}
for $\rho^{P,\theta}_\cdot(\cdot)$ defined in (\ref{eq:DefRhoPtheta}).
The above considerations then lead to a good-deal hedging problem that is analogous to (\ref{eq:GD-hedging-Problem-P-Theta}) in the case without uncertainty and, 
mathematically, for a $\F_T$-measurable contingent claim $X\in \mathbb{L}^2(\filt_+)$ writes: Find $\bar{\phi}\in\Phi$ such that for all  $P\in\Pred_{[\underline{a},\overline{a}]}$ holds $P$-almost surely
for all $t\in[0,T]$,
\begin{equation}\label{eq:GD-hedging-Problem}
\begin{split}
	\pi^u_t(X)&= \essinf^{\qquad\quad P}_{\phi\in\Phi}\rho_t\Big( X-\int_t^T\phi_s^{\text{tr}}\big(\widehat{a}^{1/2}_s\widehat{\xi}_sds+dB_s\big)\Big)\\
		  &=\rho_t\Big( X-\int_t^T\bar{\phi}_s^{\text{tr}}\big(\widehat{a}^{1/2}_s\widehat{\xi}_sds+dB_s\big)\Big).
\end{split}
\end{equation}
Analogously to \citet{BechererKentiaTonleu} 
we relate robustness of the hedging strategy (w.r.t.\ uncertainty) to a supermartingale property of its tracking (hedging) error under a class $\mathcal{P}^{\text{ngd}}$ of a-priori 
valuation measures containing all $\mathcal{P}^{\text{ngd}}(Q)$, uniformly over all reference models $Q\in\mathcal{R}$. 
To introduce in more precise terms the notion of robustness w.r.t.\ ambiguity about drifts and volatilities, we define the tracking error $R^\phi(X)$ of a strategy $\phi\in\Phi$ for a claim $X\in \mathbb{L}^2(\filt_+)$ as 
\begin{equation}
\label{eq:GD-tracking-error}
 R^\phi(X) := \pi^u_\cdot(X) - \pi^u_0(X)-\int_0^\cdot\phi_s^{\text{tr}}\big(\widehat{a}^{1/2}_s\widehat{\xi}_sds+dB_s\big). 
\end{equation}
In words, $R^\phi_t(X)$ is the difference between the dynamic variations in the (monetary) capital requirement for $X$ and the profit/loss from trading (hedging) according to $\phi$ up to 
time $t$. Note that also $\pi^u$ is a dynamic coherent risk measure. 
Subsequently, we will say that that a good-deal hedging strategy $\bar{\phi}(X)$ for a claim $X$ is 
robust w.r.t.\ (drift and volatility) uncertainty if $R^{\bar{\phi}}(X)$ is a $(\filt^P,Q)$-supermartingale for every $Q\in\mathcal{P}^{\text{ngd}}(Q^{P,\theta})$ uniformly for all 
$P\in\Pred_{[\underline{a},\overline{a}]}$ and $\theta\in\Theta$. 
\begin{remark}
If the tracking error $R^{\bar{\phi}}(X)$ is a $(\filt^P,Q)$-supermartingale for any $Q$ in $\Pred^{\text{ngd}}(Q^{P,\theta})$, then the strategy $\bar{\phi}$ is said to be 
``at least mean-self-financing'' in the model $Q^{P,\theta}$,
analogously to the property of being mean-self-financing \citep[like risk-minimizing strategies studied in][Sect.2,  with valuations taken under the minimal martingale measure]{Schweizer}
which would correspond to a martingale property (under $Q^{P,\theta}$)
of the tracking errors. Holding uniformly over all $P\in\Pred_{[\underline{a},\overline{a}]}$ and $\theta\in\Theta$, 
such can be interpreted as a robustness property of $\bar{\phi}$ \citep[cf.][]{BechererKentiaTonleu} w.r.t.\ ambiguity about drifts and volatilities.
\end{remark}
We show in Section \ref{sec:HedgingStrat2BSDE} that the robust good-deal hedging strategy $\bar{\phi}$ can be obtained in terms of the control process of a 2BSDE describing the good-deal valuation 
bound $\pi^u$, and that it is robust with respect to combined uncertainty. This strategy will be shown to be quite different from the almost-sure hedging (i.e.\ superreplicating) strategy
in general, in Section \ref{sec:ExamplewithPutOptionVolUncert}.

\subsection{Good-deal valuation bounds}\label{subsec:ValuationBounds2BSDE}

For each $P\in\Pred_{[\underline{a},\overline{a}]},\ t\in[0,T] $ and $P'\in\Pred_{[\underline{a},\overline{a}]}(t,P,\filt_+)$, the worst-case good-deal bound under drift uncertainty in 
the model $P'$ is defined for $X\in L^2(\F^{P'}_T,P')$ by 
\begin{equation*}\pi^{u,P'}_t(X):=\esssup^{\qquad\quad P'}_{\theta\in\Theta}\pi^{u,P',\theta}_t(X),\ P'\text{-a.s.},
\end{equation*}
so that (\ref{eq:DefGDBounds}) rewrites for $\F_T$-measurable $X\in\mathbb{L}^2(\filt_+)$ as
\begin{equation}\label{eq:PiuFunctionPiuPPrime}
\pi^u_t(X) = \esssup^{\qquad\quad P}_{P'\in\Pred_{[\underline{a},\overline{a}]}(t,P,\filt_+)}\pi^{u,P'}_t(X),\  t\in[0,T],\ P\text{-a.s.},\ P\in\Pred_{[\underline{a},\overline{a}]}.
\end{equation}
Note from \citep[][Thm.4.11]{BechererKentiaTonleu} that the worst-case good-deal bound under drift uncertainty $\pi^{u,P}_\cdot(X)$ for $P\in\Pred_{[\underline{a},\overline{a}]}$ is the value process of the 
standard BSDE under $P$ with terminal condition $X$ and generator $-\widehat{F}(t,\omega,\cdot) = -F(t,\omega_{\cdot\wedge t},\cdot,\widehat{a}_t(\omega))$, where
\begin{equation}\label{eq:Generator2BSDEPiU}
F(t,\omega,z,a) := \inf_{\theta\in\Theta} F^\theta(t,\omega,z,a),\quad \text{for all }(t,\omega,z,a)\in[0,T]\times\Omega\times\R^n\times\mathbb{S}^{>0}_n,
\end{equation}
and $F^\theta$ is given by (\ref{eq:GeneratorFunctionFtheta}) for each selection $\theta$ of $\Theta$.
With (\ref{eq:PiuFunctionPiuPPrime}) and the above BSDE representation of $\pi^{u,P}_\cdot(X)$, $P\in \Pred_{[\underline{a},\overline{a}]}$, the robust good-deal bound $\pi^u_\cdot(X)$ 
can alternatively be formulated like in \citep[][eq.(4.12) and Prop.4.10]{SonerTouziZhang-Dual} \citep[or alternatively][eq.(2.6) and Lem.3.5]{PossamaiTanZhou} 
as the value process of a stochastic control problem of nonlinear kernels. Roughly, this means to write, up to taking right-limits in time over rationals,
\begin{equation}\label{eq:ControlPbnonlinExp}
	\pi^u_t(X)(\omega) = \sup_{P\in\Pred_{[\underline{a},\overline{a}]}^t}\pi^{u,P,t,\omega}_t(X) \quad \text{for all } (t,\omega)\in[0,T]\times\Omega,
\end{equation}
where $\pi^{u,P,t,\omega}_\cdot(X)$ are given by solutions to standard BSDEs under $P\in \Pred_{[\underline{a},\overline{a}]}^t$, with terminal value $X$ and generator as in (\ref{eq:Generator2BSDEPiU}), 
but defined on the canonical space $\Omega^t:=\{\tilde{\omega}\in \mathcal{C}([t,T],\R^n): \tilde{\omega}(t)=0\}$ of (shifted) paths with (shifted) canonical process $B^t$, (shifted) natural filtration $\filt^t$, and associated (shifted) set of priors $\Pred_{[\underline{a},\overline{a}]}^t$, $t\in[0,T]$.
Note that by a zero-one law as in Lemma \ref{lem:QSAggregPS}, $\pi^{u,P,t,\omega}_t(X)$ is indeed constant for any $(t,\omega)$ and $P\in \Pred_{[\underline{a},\overline{a}]}^t$ (like in Remark \ref{rem:DynProgPrinciple}, Part 2.), 
and hence the pointwise supremum (\ref{eq:ControlPbnonlinExp}) is well-defined. Although such a pathwise description of the worst-case good-deal bound $\pi^u_t(X)$ reflects better the economic intuition behind the latter
and, as a stochastic control problem, would be more classical for the literature, it might be less suitable for approaching the hedging problem (\ref{eq:GD-hedging-Problem}) for which the 
essential supremum formulation (\ref{eq:PiuFunctionPiuPPrime}) seems more appropriate. 

It is known \citet{SonerTouziZhang-Wellposedness,PossamaiTanZhou} that under suitable assumptions (on terminal condition and generator) 
the value process of a control problem like (\ref{eq:ControlPbnonlinExp}) can be described in terms of the solution to a 2BSDE. We provide such a description for $\pi^u_\cdot(X)$ by considering the 2BSDE
\begin{equation}\label{eq:2BSDEpiU}
  Y_t = X-\phantom{}^{(P)}\hspace{-0.1cm}\int_t^TZ_s^{\text{tr}}dB_s - \int_t^T\widehat{F}_s(\widehat{a}^{1/2}_sZ_s)ds +K^P_T-K^P_t,\ t\in[0,T],\ \Pred_{[\underline{a},\overline{a}]}\text{-q.s.},
  \end{equation}
with generator $F$ given by (\ref{eq:Generator2BSDEPiU}) and terminal condition $X$. We have the following 
\begin{theorem}[Good-deal valuation under combined uncertainty]\phantomsection{}\label{thm:2BSDECHaractforPiU}
Let $X$ be a $\F_T$-measurable claim in $\mathbb{L}^2(\filt_+)$. Then a unique solution $(Y,Z,(K^P)_{P\in\Pred_{[\underline{a},\overline{a}]}})\in \mathbb{D}^2\big(\filt^{\Pred_{[\underline{a},\overline{a}]}}\big)\times \mathbb{H}^2\big(\filt^{\Pred_{[\underline{a},\overline{a}]}}\big)\times \mathbb{I}^2\big(\big(\filt^P\big)_{P\in\Pred_{[\underline{a},\overline{a}]}}\big)$ to the 2BSDE (\ref{eq:2BSDEpiU}).
Moreover the process $Y$ uniquely satisfies (\ref{eq:DefGDBounds}) and any $P\in\Pred_{[\underline{a},\overline{a}]}$ holds $\pi^u_t(X)=Y_t,\ t\le T,\ P\text{-a.s.}$. 
\end{theorem}
\begin{proof}
We aim to apply Part 1.\ of Proposition \ref{pro:ExistenceUniquenessSol2BSDEs} to show existence and uniqueness of the solution to the 2BSDE (\ref{eq:2BSDEpiU}). To this end and since $X$ is $\F_T$-measurable and in $\mathbb{L}^2(\filt_+)$ 
it suffices to check that $X$ and the function $F$ satisfy conditions (ii) to (iv) of Assumption \ref{asp:Assumption1}.
Because $\widehat{F}^0\equiv0$, then (iv) obviously holds. As for (iii) about the uniform Lipschitz continuity of $F$ in $z$, it is enough to show that the functions
$F^{\,\theta}$, for $\theta\in\Theta$, are equi-Lipschitz in $z$ uniformly for all $(t,\omega,a)\in[0,T]\times\Omega\times \mathbb{S}^{>0}_n$.
The latter holds since by Minkowski and Cauchy-Schwarz inequalities one has 
\begin{align*}
& \big\lvert F^{\,\theta}(t,\omega,z,a) - F^{\,\theta}(t,\omega,z',a)\big\rvert \\
&\qquad \le \big\lvert\Pi^{a,\bot}_{(t,\omega)}\big(\theta_t(\omega)\big)\big\rvert\cdot\big\lvert\Pi^{a,\bot}_{(t,\omega)}(z-z')\big\rvert+\big\lvert\widehat{\xi}_t(\omega)\big\rvert\cdot\big\lvert \Pi^a_{(t,\omega)}(z-z')\big\rvert\\
&\qquad\ \quad + \big(h^2_t(\omega)-\big\lvert \widehat{\xi}_t(\omega)+\Pi^a_{(t,\omega)}(\theta_t(\omega))\big\rvert^2\big)^{\frac{1}{2}} \big\lvert\Pi^{a,\bot}_{(t,\omega)}(z-z')\big\rvert\\
								     &\qquad  \le \delta_t(\omega)\big\lvert z-z'\big\rvert +\big\lvert\widehat{\xi}_t(\omega)\big\rvert\big\lvert z-z'\big\rvert + h_t(\omega) \big\lvert z-z'\big\rvert
								      \le C\big\lvert z-z'\big\rvert, 
\end{align*}
for all $\theta\in\Theta$, for $C\in(0,\infty),$ making use of (\ref{eq:AspOnxitheta}) and boundedness of the functions $\delta$ and $h$. 
To infer the first claim of the theorem, it remains to show (ii) about $F$ being jointly and $\filt$-progressive. 
We first show that $F$ is $\filt$-progressive in $(t,\omega)$. For each $z\in\R^n,\ a\in \mathbb{S}^{>0}_n$, the map $[0,T]\times\Omega\times\R^n\ni(t,\omega,\vartheta)\mapsto F^\vartheta(t,\omega,z,a)$ defined from (\ref{eq:GeneratorFunctionFtheta}), after straightforward calculations by (omitting the dependence on $\omega$ for notational simplicity)
\begin{align}\nonumber
	  & F^\vartheta(t,\cdot,z,a) =  -\vartheta^{\text{tr}}z+\vartheta^{\text{tr}}a^{\frac{1}{2}}\sigma_t^{\text{tr}}(\sigma_ta\sigma_t^{\text{tr}})^{-1}\sigma_ta^{\frac{1}{2}}z +\widehat{\xi}^{\,\text{tr}}_ta^{\frac{1}{2}}\sigma_t^{\text{tr}}(\sigma_ta\sigma_t^{\text{tr}})^{-1}\sigma_ta^{\frac{1}{2}}z \\
		\; \label{eq:GeneratorFunctionFthetaExplicit}
			      &-\big(h^2_t - \big\lvert \widehat{\xi}_t+ a^{\frac{1}{2}}\sigma_t^{\text{tr}}(\sigma_ta\sigma_t^{\text{tr}})^{-1}\sigma_ta^{\frac{1}{2}}\vartheta\big\rvert^2\big)^{\frac{1}{2}}
					      \big(\lvert z\rvert^2 - z^{\text{tr}}a^{\frac{1}{2}}\sigma_t^{\text{tr}}(\sigma_ta\sigma_t^{\text{tr}})^{-1}\sigma_ta^{\frac{1}{2}}z \big)^{\frac{1}{2}}
\end{align}
is continuous in $\vartheta$ and $\filt$-predictable in $(t,\omega)$ since $\sigma, h,\widehat{\xi}$
are $\filt$-predictable. Since the correspondence $\Theta$ is $\filt$-predictable then by measurable maximum and measurable selection results \citep[][Thms 2K and 1.C]{Rockafellar} one has for each $(z,a)\in \R^n\times\mathbb{S}^{>0}_n$
that $F(\cdot,\cdot,z,a)$ is $\filt$-predictable (hence $\filt$-progressive) and there exists $\theta^{z,a}\in\Theta$ such that $F(t,\omega,z,a) = F^{\theta^{z,a}(t,\omega)}(t,\omega,z,a)=\inf_{\vartheta\in\Theta_t(\omega)} F^\vartheta(t,\omega,z,a)$ for all $(t,\omega)\in[0,T]\times\Omega$. 
Now consider $\widetilde{\Theta}:[0,T]\times\Omega\times\R^n\times\mathbb{S}^{>0}_n\rightsquigarrow \R^n$, a correspondence that is constant in its last two arguments and 
defined by $\widetilde{\Theta}_t(\omega,z,a):=\Theta_t(\omega)$
for all $(t,\omega,z,a)\in[0,T]\times\Omega\times\R^n\times\mathbb{S}^{>0}_n.$ Then $\widetilde{\Theta}$ is a $\Pred(\filt)\otimes \mathcal{B}(\R^n)\otimes\mathcal{B}(\mathbb{S}^{>0}_n)$-measurable correspondence because $\Theta$ is $\filt$-predictable,
with $\Pred(\filt)$ denoting the predictable sigma-field w.r.t.\ $\filt$. To prove joint measurability of $F$, it suffices by measurable selection arguments analogous to the ones above 
to show that the map $(t,\omega,z,a)\mapsto \inf_{\vartheta\in\widetilde{\Theta}_t(\omega,z,a)} F^\vartheta(t,\omega,z,a)$ is $\mathcal{B}([0,T])\otimes\F_T\otimes\mathcal{B}(\R^n)\otimes\mathcal{B}(\mathbb{S}^{>0}_n)$-measurable.
Showing that $(t,\omega,z,a,\vartheta)\mapsto F^\vartheta(t,\omega,z,a)$ is a Carath\'eodory function would imply the result by \citep[][Thm.2K]{Rockafellar}. 
As $\vartheta\mapsto F^\vartheta(t,\omega,z,a)$ is continuous for each $(t,\omega,z,a)$, it remains only to show that $(t,\omega,z,a)\mapsto F^\vartheta(t,\omega,z,a)$ is jointly measurable for each $\vartheta$.
First note that $F^\vartheta(\cdot,\cdot,z,a)$ is $\mathcal{B}([0,T])\otimes\F_T$-measurable for each $(z,a,\vartheta)\in\R^n\times\mathbb{S}^{>0}_n\times\R^n$ since it is $\filt$-predictable. 
Moreover for any $x\in\R^{n\times k},y\in\R^{n\times p}$ (with $k,p\in\N$) the function $\R^{n\times n}\ni a\mapsto x^{\text{tr}}ay\in\R^{k\times p}$ is continuous. Likewise $a\mapsto a^{-1}$ and $a\mapsto a^{1/2}$ are continuous over $\mathbb{S}^{>0}_n$,
the former being a consequence of the continuity of the determinant operator. 
Hence from (\ref{eq:GeneratorFunctionFthetaExplicit}) one infers that the map $\mathbb{S}^{>0}_n\ni a\mapsto F^\vartheta(t,\omega,z,a)$ is continuous. Overall because $F^\vartheta$ 
is Lipschitz in $z$ uniformly in $(t,\omega,a)$ and continuous in $a$, it is continuous in $(z,a)$. As a Carath\'eodory function, 
$(t,\omega,z,a)\mapsto F^\vartheta(t,\omega,z,a)$ is jointly measurable and this concludes 
joint measurability of $F$. Overall we have shown that $F$ satisfies Assumption \ref{asp:Assumption1}. Hence the first claim of the theorem follows by Part 1.\ of Proposition \ref{pro:ExistenceUniquenessSol2BSDEs} while  
the second claim is obtained from Part 2.\ of Proposition \ref{pro:ExistenceUniquenessSol2BSDEs}, recalling the expression (\ref{eq:PiuFunctionPiuPPrime}) for the good-deal bound 
and using \citep[][Thm.4.11]{BechererKentiaTonleu}. 
\end{proof}
\begin{remark}\phantomsection{}\label{rem:DynProgPrinciple}
  1.\  
In a situation with only volatility uncertainty and zero drift, one could apply instead of \citet{PossamaiTanZhou} an earlier but less general 
  wellposedness result by \citet{SonerTouziZhang-Wellposedness} to the 2BSDE (\ref{eq:2BSDEpiU}),
  as noted in \citep[][Thm.4.20]{KentiaPhD}. 
 The latter requires (besides some continuity of $X$) the generator function $F$ to be convex (in $a$) and uniformly continuous (UC)
  in $\omega$, which would hold under UC-assumptions on $\sigma,h$. However as soon as one considers non-zero drift (let alone drift uncertainty), convexity of $F$ from our application is 
  no longer clear. Uniform continuity for $F$ given by (\ref{eq:Generator2BSDEPiU}) in the present situation with combined uncertainties is neither. E.g., a sufficient condition for it would be that the family 
  $(F^\theta)_{\theta\in\Theta}$ is equicontinuous (in $\omega$), what however seems restrictive.
  
  2.\ Theorem \ref{thm:2BSDECHaractforPiU} shows in particular that the family of essential suprema in (\ref{eq:DefGDBounds}) indexed by the measures $P\in\Pred_{[\underline{a},\overline{a}]}$ effectively aggregates into a single 
  process $\pi^u_\cdot(X)\in \mathbb{D}^2\big(\filt^{\Pred_{[\underline{a},\overline{a}]}}\big)$ for Borel-measurable claims $X\in \mathbb{L}^2(\filt_+)$. In this case $\pi^u_\cdot(X)$
  is $\filt^{\Pred_{[\underline{a},\overline{a}]}}$-progressive, and therefore $\pi^u_0(X)$ is $\F^{\Pred_{[\underline{a},\overline{a}]}}_0-$measurable. 
   Hence $\pi^u_0(X)$ is deterministic constant by the zero-one law of Lemma \ref{lem:QSAggregPS}, with $\pi^u_0(X) = \sup_{P\in \Pred_{[\underline{a},\overline{a}]}}\pi^{u,P}_0(X)$.

  3.\  By Proposition \ref{pro:ExistenceUniquenessSol2BSDEs}, the good-deal bound $\pi^u_\cdot(\cdot)$ satisfies a dynamic programing principle  for each Borel-measurable claim $X$ in $\mathbb{L}^2(\filt_+)$: 
   for any $P\in\Pred_{[\underline{a},\overline{a}]}$ holds
  \begin{equation*}\displaystyle \pi^u_s(X)= \esssup^{\qquad\quad P}_{P'\in\Pred_{[\underline{a},\overline{a}]}(s,P,\filt_+)} \pi^{u,P'}_s(\pi^u_t(X))
             = \pi^u_s(\pi^u_t(X)),\ s\le t\le T,\ P\text{-a.s.}.
   \end{equation*}
   
  4.\ Like in \citep[][Thm.2.7]{KloppelSchweizer} 
  (cf.\ \citealt[][Prop.2.6]{Becherer-Good-Deals} or \citealt[][Lem.2.1, Part a)]{BechererKentiaTonleu}) 
  one can show that  $(t,X)\mapsto\pi^u_t(X)$ has the properties of a time-consistent dynamic coherent risk measure. 
\end{remark}
\subsection{Robust good-deal hedging}\label{sec:HedgingStrat2BSDE}
We now investigate robust good-deal hedging strategies relying on the 2BSDE theory of Section \ref{sec:2BSDEs}. 
By (\ref{eq:DefRhoFirst}) we have for $X\in L^2_{\Pred_{[\underline{a},\overline{a}]}}$ that 
\begin{equation}\label{eq:EqRho}
 \rho_t(X)= \esssup^{\qquad\quad P}_{P'\in\Pred_{[\underline{a},\overline{a}]}(t,P,\filt_+)} \rho^{P'}_t(X),\ t\in[0,T],\ P\text{-a.s.\  for all }P\in\Pred_{[\underline{a},\overline{a}]},
\end{equation}
\begin{equation*}
\text {where } \quad \rho^{P}_t(X):=\esssup^{\qquad\quad P}_{\theta\in\Theta}\rho^{P,\theta}_t(X),\ t\in[0,T].
\end{equation*}
As the set $\mathcal{P}^{\text{ngd}}(Q^{P,\theta})$ is m-stable and convex for each $P\in\Pred_{[\underline{a},\overline{a}]},\ \theta\in\Theta$, then by \citep[][Lem.4.9]{BechererKentiaTonleu} the union $\bigcup_{\theta\in\Theta}\mathcal{P}^{\text{ngd}}(Q^{P,\theta})$ is also m-stable and convex. 
Thus the dynamic coherent risk measure $\rho^P:L^2(P)\to L^2(P,\F_t)$ is time-consistent \citep[see e.g.][Lem.2.1]{BechererKentiaTonleu}, with $\rho^{P}_t(X)\ge\pi^{u,P}_t(X),\ t\in[0,T],\ P$-a.s.\ for all $X\in L^2(P,\F_T)$.
Consider the generator function $F':[0,T]\times\Omega\times\R^n\times\mathbb{S}^{>0}_n\rightarrow \R$ defined by 
\begin{equation}\label{eq:DefforRho}
F'(t,\omega,z,a) := -(h_t+\delta_t)\lvert z\rvert,\quad \text{for all } (t,\omega,z,a)\in[0,T]\times\Omega\times\R^n\times\mathbb{S}^{>0}_n, 
\end{equation}
and the associated 2BSDE
\begin{equation}\label{eq:2BSDEYPRime}
	  Y'_t = X-\phantom{}^{(P)}\hspace{-0.1cm}\int_t^T{Z'}_s^{\text{tr}}dB_s - \int_t^T\widehat{F}'_s(\widehat{a}^{1/2}_sZ'_s)ds +K'^{P}_T-K'^{P}_t,\ t\in[0,T],\ \Pred_{[\underline{a},\overline{a}]}\text{-q.s.}.
\end{equation}
Proposition \ref{Pro:RepresentationRhoX} below gives a 2BSDE description of $\rho_\cdot(X)$ that, similar to parts 3.-4.\ of Remark \ref{rem:DynProgPrinciple}  
for $\pi^u_\cdot(\cdot)$, implies that $\rho$ defines a dynamic risk measure (analogous to $\pi^u_\cdot(\cdot)$) that is time-consistent over 
all $\F_T$-measurable $X$ in $\mathbb{L}^2(\filt_+)$. The proof is analogous to that of Theorem \ref{thm:2BSDECHaractforPiU}, and its details are included in the appendix. 
\begin{proposition}\phantomsection{}\label{Pro:RepresentationRhoX}
Let $X$ be a $\F_T$-measurable claim in $\mathbb{L}^2(\filt_+)$. Then there exists a unique solution $(Y',Z',(K'^{P})_{P\in\Pred_{[\underline{a},\overline{a}]}})\in \mathbb{D}^2\big(\filt^{\Pred_{[\underline{a},\overline{a}]}}\big)\times \mathbb{H}^2\big(\filt^{\Pred_{[\underline{a},\overline{a}]}}\big)\times \mathbb{I}^2\big(\big(\filt^P\big)_{P\in\Pred_{[\underline{a},\overline{a}]}}\big)$ to the 2BSDE (\ref{eq:2BSDEYPRime}).
Moreover the process $Y'$ uniquely satisfies (\ref{eq:EqRho}) and for all $P\in\Pred_{[\underline{a},\overline{a}]}$ holds $\rho_t(X)=Y'_t,\ t\in[0,T],\ P\text{-a.s.}$.
\end{proposition}
Note from results on good-deal valuation and hedging in the presence solely of drift uncertainty \citep[cf.\ e.g.][Prop.4.12 and Thm.4.13]{BechererKentiaTonleu}, that 
for any $P\in\Pred_{[\underline{a},\overline{a}]},\ P'\in\Pred_{[\underline{a},\overline{a}]}(t,P,\filt_+)$ and $X\in L^2(\F_T^P, P)$ one has, $\ P\text{-a.s.}$  on $\ t\in[0,T]$,
 \begin{equation*}
  \pi^{u,P'}_t(X) =  \essinf^{\qquad\quad P}_{\phi\in\Phi(P')}\rho^{P'}_t\Big( X-\phantom{}^{(P')}\hspace{-0.1cm}{\int_t^T}\phi_s^{\text{tr}}\big(\widehat{a}^{1/2}_s\widehat{\xi}_sds+dB_s\big)\Big).
 \end{equation*}
Hence from (\ref{eq:PiuFunctionPiuPPrime}) and for $\F_T$-measurable claims $X$ in $\mathbb{L}^2(\filt_+)$ holds for all $P\in\Pred_{[\underline{a},\overline{a}]}$, $P$-almost surely, 
\begin{equation}\label{eq:PiUStandard}
 \pi^u_t(X)= \esssup^{\qquad\quad P}_{P'\in\Pred_{[\underline{a},\overline{a}]}(t,P,\filt_+)} \essinf^{\qquad\quad P}_{\phi\in\Phi(P')}\rho^{P'}_t\Big( X-\phantom{}^{(P')}\hspace{-0.1cm}\int_t^T\phi_s^{\text{tr}}\big(\widehat{a}^{1/2}_s\widehat{\xi}_sds+dB_s\big)\Big),\ t\in[0,T].
\end{equation}
Furthermore one can infer from \citep[][Thms.\ 4.11, 4.13]{BechererKentiaTonleu} the existence of  
a family $\left\lbrace \bar{\phi}^P\in\Phi(P),\ P\in\Pred_{[\underline{a},\overline{a}]}\right\rbrace$ of trading strategies satisfying for each $P\in\Pred_{[\underline{a},\overline{a}]},$
\begin{equation}\label{eq:KindOfSolHedgProb}
 \pi^u_t(X)= \esssup^{\qquad\quad P}_{P'\in\Pred_{[\underline{a},\overline{a}]}(t,P,\filt_+)} \rho^{P'}_t\Big( X-\phantom{}^{(P')}\hspace{-0.1cm}\int_t^T(\bar{\phi}^{P'}_s)^{\text{tr}}\big(\widehat{a}^{1/2}_s\widehat{\xi}_sds+dB_s\big)\Big),\ t\in[0,T],\ P\text{-a.s.}.
\end{equation}
In addition $\bar{\phi}^P$ is given for each $P\in\Pred_{[\underline{a},\overline{a}]}$ by
\begin{equation*}
 \widehat{a}^{\frac{1}{2}}_t\,\bar{\phi}^P_t = \widehat{\Pi}_t\big(\widehat{a}_t^{\frac{1}{2}}\mathcal{Z}^{P,X}_t\big)+ \frac{\big\lvert \Pi^\bot_t\big(\widehat{a}_t^{\frac{1}{2}}\mathcal{Z}^{P,X}_t\big)\big\rvert}{\sqrt{h^2_t-\big\lvert\widehat{\xi}_t+\widehat{\Pi}_t\big(\bar{\theta}^P_t\big)\big\rvert^2}}\Big(\widehat{\xi}_t+\widehat{\Pi}_t\big(\bar{\theta}^P_t\big)\Big),\ P\otimes dt\text{-a.e.}, 
\end{equation*}
where $(\mathcal{Y}^{P,X},\mathcal{Z}^{P,X})$ with $\mathcal{Y}^{P,X} =\pi^{u,P}_\cdot(X)$ is the solution to the standard BSDE under $P\in\Pred_{[\underline{a},\overline{a}]}$ with 
data $(-\widehat{F}(\widehat{a}^{\frac{1}{2}}\cdot),X)$ for $F$ defined in (\ref{eq:Generator2BSDEPiU}), and $\bar{\theta}^P$ is a $\filt^P$-predictable selection of $\Theta$ satisfying  
\begin{equation*}
\widehat{F}_t(\widehat{a}_t^{\frac{1}{2}}\mathcal{Z}^{P,X}_t) = \widehat{F}^{\bar{\theta}^P}_t(\widehat{a}_t^{\frac{1}{2}}\mathcal{Z}^{P,X}_t)=F^{\bar{\theta}^P}(t,\widehat{a}_t^{\frac{1}{2}}\mathcal{Z}^{P,X}_t,\widehat{a}_t)\quad P\otimes dt\text{-a.e.},
\end{equation*}
with $F^{\bar{\theta}^P}$ given as in (\ref{eq:GeneratorFunctionFtheta}) for $\theta=\bar{\theta}^P$. 
\begin{remark}
In the case of no volatility uncertainty \citep[as studied in][]{BechererKentiaTonleu}, i.e.\ for $\Pred_{[\underline{a},\overline{a}]}=\{P^0\}$ with $\underline{a}=\overline{a}=I_{n\times n}$,
the strategy $\bar{\phi}^{P^0}(X)$ for a claim $X$ in $L^2_{\Pred_{[\underline{a},\overline{a}]}}(\F_T)=\mathbb{L}^2(\filt_+)=L^2(\F_T,P^0)$ would be $\Pred_{[\underline{a},\overline{a}]}$-permitted and hence already the robust good-deal hedging strategy (w.r.t.\ drift uncertainty solely) 
for robust valuation $\pi^u_\cdot(X)=\pi^{u,P^0}_\cdot(X)$ and risk measure $\rho=\rho^{P^0}$. In the presence of volatility uncertainty however, the situation is more complicated because each 
strategy $\bar{\phi}^P$ and risk measure $\rho^P$ may be defined only up to a null-set of each (non-dominated) measure $P\in\Pred_{[\underline{a},\overline{a}]}$.
\end{remark}
Since we are looking for a single process $\bar{\phi}\in\Phi$ solution to the hedging problem  (\ref{eq:GD-hedging-Problem}), one way is to 
investigate appropriate conditions under which the family $\{\bar{\phi}^P,\ P\in\Pred_{[\underline{a},\overline{a}]}\}$ can be aggregated into a single strategy $\bar{\phi}\in\Phi$, i.e.\ with $\bar{\phi} = \bar{\phi}^P\ P\otimes dt\text{-a.e.}$, for all $P\in\Pred_{[\underline{a},\overline{a}]}$.
Were this possible, (\ref{eq:KindOfSolHedgProb}) would be written for any $P\in\Pred_{[\underline{a},\overline{a}]}$ as
\begin{align*}
 \pi^u_t(X)&= \esssup^{\qquad\quad P}_{P'\in\Pred_{[\underline{a},\overline{a}]}(t,P,\filt_+)} \rho^{P'}_t\Big( X-\int_t^T\bar{\phi}_s^{\text{tr}}\big(\widehat{a}^{1/2}_s\widehat{\xi}_sds+dB_s\big)\Big),\ t\in[0,T].\ P\text{-a.s.} \\
           &= \rho_t\Big( X-\int_t^T\bar{\phi}_s^{\text{tr}}\big(\widehat{a}^{1/2}_s\widehat{\xi}_sds+dB_s\big)\Big),\ t\in[0,T],\ P\text{-a.s.},
\end{align*}
such that the aggregate $\bar{\phi}$ would readily satisfy (\ref{eq:GD-hedging-Problem}) since $\pi^u_\cdot(\cdot)\le \rho_\cdot(\cdot)$.
However, since general conditions for aggregation are somewhat restrictive and highly technical \citep[see e.g.][]{SonerTouziZhang-Aggregation}, we shall abstain from following this path and will directly
show that a suitable strategy described in terms of the $Z$-component of a 2BSDE (cf.\ (\ref{eq:DefBarPhi}) below) solves the good-deal hedging problem.
In the simpler framework of drift uncertainty only, there exists a 
worst-case measure $\bar{P}\in \mathcal{R}$ such that $\pi^u_\cdot(\cdot)=\pi^{u,\bar{P}}_\cdot(\cdot)$ holds and the good-deal
hedging strategy $\bar{\phi}^{\bar{P}}$ in the model $\bar{P}$ is robust w.r.t.\ (drift) uncertainty. This has been shown in 
\citet{BechererKentiaTonleu} by  
first considering an auxiliary (larger) valuation bound for which the associated hedging strategy automatically satisfies the robustness property
and then by a saddle point result \citep[cf.][Thm.4.13]{BechererKentiaTonleu} by identifying the auxiliary bound and hedging strategy with the standard good-deal bound and hedging 
strategy respectively. 
In the present non-dominated framework, however, it is questionable but indeed rather unlikely that such a worst-case measure $\bar{P}$ exists in $\mathcal{R}$ for general claims. 
Instead, we shall follow a slightly different approach and show in a more straightforward manner using Lemma \ref{lem:SaddePointOptPb} that the candidate strategy in (\ref{eq:DefBarPhi}) 
indeed satisfies the required robustness property w.r.t.\ drift and volatility uncertainty, which will then be used to recover its optimality. 
Note that our proof does not use any comparison theorem for 2BSDEs \citep[as e.g.\ in][Sect.4.1.3]{KentiaPhD}. Such would need the terminal wealths $\int_0^T\phi^{\text{tr}}_s dB_s$ 
as possible terminal conditions for 2BSDEs to be in $\mathbb{L}^2(\filt_+)$, which might not be the case for $\phi\in\Phi$.

For a Borel-measurable claim $X$ in $\mathbb{L}^2(\filt_+)$, consider the process $\bar{\phi} := \bar{\phi}(X)$ defined by 
\begin{equation}\label{eq:DefBarPhi}
 \widehat{a}^{\frac{1}{2}}_t\bar{\phi}_t := \widehat{\Pi}_t\big(\widehat{a}_t^{\frac{1}{2}}Z_t\big)+ \frac{\big\lvert \Pi^\bot_t\big(\widehat{a}_t^{\frac{1}{2}}Z_t\big)\big\rvert}{\sqrt{h^2_t-\big\lvert\widehat{\xi}_t+\widehat{\Pi}_t\big(\bar{\theta}_t\big)\big\rvert^2}}\Big(\widehat{\xi}_t+\widehat{\Pi}_t\big(\bar{\theta}_t\big)\Big),\ \Pred_{[\underline{a},\overline{a}]}\otimes dt\text{-q.e.},
\end{equation}
for the unique solution $(Y,Z,(K^P)_{P\in\Pred_{[\underline{a},\overline{a}]}})$ to the 2BSDE (\ref{eq:2BSDEpiU}), $\bar{\theta}$ being a $\filt^{\Pred_{[\underline{a},\overline{a}]}}$-predictable process with $\lvert \bar{\theta}_t(\omega)\rvert\le \delta_t(\omega)$ for all $(t,\omega)\in[0,T]\times\Omega$                                                              
and  satisfying 
\begin{equation}\label{eq:OptEquaThetaBar}
\widehat{F}_t(\widehat{a}_t^{\frac{1}{2}}Z_t) = \widehat{F}^{\bar{\theta}}_t(\widehat{a}_t^{\frac{1}{2}}Z_t)=F^{\bar{\theta}}(t,\widehat{a}_t^{\frac{1}{2}}Z_t,\widehat{a}_t),\ \Pred_{[\underline{a},\overline{a}]}\otimes dt\text{-q.e.}.
\end{equation}
Existence of $\bar{\theta}\in\Theta$ satisfying (\ref{eq:OptEquaThetaBar}) easily follows by measurable selection arguments similar to those in the proof of Theorem \ref{thm:2BSDECHaractforPiU},
using the fact that $Z$ is $\filt^{\Pred_{[\underline{a},\overline{a}]}}$-predictable, $\Theta$ is $\filt$-predictable and $\Pred(\filt)\subseteq \Pred\big(\filt^{\Pred_{[\underline{a},\overline{a}]}}\big).$
The following result shows that $\bar{\phi}$ given by (\ref{eq:DefBarPhi}) is indeed a robust good-deal hedging strategy if the family $\big\lbrace\phantom{}^{(P)}\hspace{-0.1cm}\int_0^\cdot Z_s^{\text{tr}}dB_s,\, P\in\Pred_{[\underline{a},\overline{a}]}\big\rbrace$ of ``profit$\, \&\, $loss'' processes aggregates.

\begin{theorem}[Robust good-deal hedging under combined uncertainty]\label{thm:GDHedgingTheorem}
	For an $\F_T$-measurable $X$ in $\mathbb{L}^2(\filt_+)$, let $(Y,Z,(K^P)_{P\in\Pred_{[\underline{a},\overline{a}]}})$ be the unique solution in $\mathbb{D}^2\big(\filt^{\Pred_{[\underline{a},\overline{a}]}}\big)\times \mathbb{H}^2\big(\filt^{\Pred_{[\underline{a},\overline{a}]}}\big)\times \mathbb{I}^2\big(\big(\filt^P\big)_{P\in\Pred_{[\underline{a},\overline{a}]}}\big)$ 
to the 2BSDE (\ref{eq:2BSDEpiU}) 
	and assume that the family $\big\lbrace\phantom{}^{(P)}\hspace{-0.1cm}\int_0^\cdot Z_s^{\text{tr}}dB_s,\, P\in\Pred_{[\underline{a},\overline{a}]}\big\rbrace$ 
	aggregates into one process $\int_0^\cdot Z_s^{\text{tr}}dB_s$. Then:
	
	1.\ The process $\bar{\phi}=\bar{\phi}(X)$ from (\ref{eq:DefBarPhi}) is in $\Phi$ and solves the good-deal hedging problem  (\ref{eq:GD-hedging-Problem}).
	
	2.\  The tracking error $R^{\bar{\phi}}(X)$ of the  good-deal hedging strategy $\bar{\phi}$  (cf. \eqref{eq:GD-tracking-error})  is a 
$(\filt^P,Q)$-supermartingale under any $Q\in\mathcal{P}^{\text{ngd}}\big(Q^{P,\theta}\big)$,
 for all $P\in\Pred_{[\underline{a},\overline{a}]}$, $\theta\in\Theta$.
\end{theorem}

Note that, apart from measurability and some integrability, no regularity conditions (like e.g.\ uniform continuity) are imposed on the contingent claim $X$. As the setup is non-Markovian, the 
contingent claim could clearly be path-dependent. 

\begin{proof}
We first prove the second claim, and then use it to imply the first. Note that by the condition on the integral $Z$, the strategy $\bar{\phi}$ given by (\ref{eq:DefBarPhi}) clearly belongs to $\Phi$. 
By Theorem \ref{thm:2BSDECHaractforPiU}, we know that $\pi^u_\cdot(X)=Y$ for $(Y,Z,(K^P)_{P\in\Pred_{[\underline{a},\overline{a}]}})$ solution to the 2BSDE (\ref{eq:2BSDEpiU}). 
Let $P\in\Pred_{[\underline{a},\overline{a}]}$, $\theta\in\Theta$ and $Q\in \mathcal{P}^{\text{ngd}}(Q^{P,\theta})$. Then $Q$ is equivalent to $Q^{P,\theta}$ and $dQ = \phantom{}^{(P)}\hspace{-0.05cm}\mathcal{E}(\lambda\cdot W^{P,\theta})dQ^{P,\theta}$ for $\lvert\lambda\rvert\le h\ P\otimes dt$-a.e..
The dynamics of $R^{\bar{\phi}}:=R^{\bar{\phi}}(X)$ is then given $P$-almost surely for all $t\in[0,T]$ by 
\begin{align*}
	-dR^{\bar{\phi}}_t &=  -\widehat{F}_t(\widehat{a}_t^{\frac{1}{2}}Z_t)dt -Z_t^{\text{tr}}dB_t + \bar{\phi}_t^{\text{tr}}\big(\widehat{a}^{\frac{1}{2}}_t\widehat{\xi}_tdt+dB_t\big)+dK^P_t\\
			   &=  \Big(\bar{\phi}_t^{\text{tr}}\widehat{a}^{\frac{1}{2}}_t\widehat{\xi}_t-\widehat{F}_t(\widehat{a}_t^{\frac{1}{2}}Z_t)\Big)dt -\big(Z_t-\bar{\phi}_t\big)^{\text{tr}}\widehat{a}^{\frac{1}{2}}_tdW^P_t+dK^P_t.
\end{align*}
By a change of measures from $P$ to $Q$ for the $Q$-Brownian motion $W^Q=W^P-\int_0^\cdot(\lambda_t+\theta_t)dt$ one obtains $P$-almost surely for all $t\in[0,T]$ that 
\begin{equation*}
	-dR^{\bar{\phi}}_t =  \Big(\bar{\phi}_t^{\text{tr}}\widehat{a}^{\frac{1}{2}}_t\widehat{\xi}_t-(\lambda_t+\theta_t)^{\text{tr}}\widehat{a}^{\frac{1}{2}}_t\big(Z_t-\bar{\phi}_t\big)-\widehat{F}_t(\widehat{a}_t^{\frac{1}{2}}Z_t)\Big)dt -\big(Z_t-\bar{\phi}_t\big)^{\text{tr}}\widehat{a}^{\frac{1}{2}}_tdW^Q_t+dK^P_t.
\end{equation*}
Since $\max_{\lvert \lambda\rvert\le h} \lambda^{\text{tr}}_t \widehat{a}_t^{\frac{1}{2}}(Z_t-\bar{\phi}_t)= h_t\big\lvert \widehat{a}_t^{\frac{1}{2}}(Z_t-\bar{\phi}_t)\big\rvert,\ P\otimes dt$-a.e., then
\begin{equation}\label{eq:hdomilambdalessh}
	\bar{\phi}_t^{\text{tr}}\widehat{a}^{\frac{1}{2}}_t\widehat{\xi}_t-(\lambda_t+\theta_t)^{\text{tr}}\widehat{a}^{\frac{1}{2}}_t\big(Z_t-\bar{\phi}_t\big)\ge \bar{\phi}_t^{\text{tr}}\widehat{a}^{\frac{1}{2}}_t\widehat{\xi}_t-\theta_t^{\text{tr}}\widehat{a}^{\frac{1}{2}}_t\big(Z_t-\bar{\phi}_t\big)- h_t\big\lvert\widehat{a}^{\frac{1}{2}}_t\big(Z_t-\bar{\phi}_t\big)\big\rvert,\ P\otimes dt\text{-a.e.}.
\end{equation}
In addition by Parts 1.\ and 2.\ of Lemma \ref{lem:SaddePointOptPb} one obtains from the definition (\ref{eq:Generator2BSDEPiU}) of $F$ and the expression (\ref{eq:DefBarPhi}) of $\bar{\phi}$ that 
$
\widehat{F}_t(\widehat{a}_t^{\frac{1}{2}}Z_t) = \bar{\phi}_t^{\text{tr}}\widehat{a}^{\frac{1}{2}}_t\widehat{\xi}_t-\bar{\theta}_t^{\text{tr}}\widehat{a}^{\frac{1}{2}}_t\big(Z_t-\bar{\phi}_t\big)- h_t\big\lvert\widehat{a}^{\frac{1}{2}}_t\big(Z_t-\bar{\phi}_t\big)\big\rvert,\ \Pred_{[\underline{a},\overline{a}]}\otimes dt\text{-q.e.},
$
for $\bar{\theta}\in\Theta$ satisfying (\ref{eq:OptEquaThetaBar}). As a consequence, part 3.\ of Lemma \ref{lem:SaddePointOptPb} yields 
\begin{equation}\label{eq:ForFinVarDec}
	 \bar{\phi}_t^{\text{tr}}\widehat{a}^{\frac{1}{2}}_t\widehat{\xi}_t-\theta_t^{\text{tr}}\widehat{a}^{\frac{1}{2}}_t\big(Z_t-\bar{\phi}_t\big)- h_t\big\lvert\widehat{a}^{\frac{1}{2}}_t\big(Z_t-\bar{\phi}_t\big)\big\rvert\ge \widehat{F}_t(\widehat{a}_t^{\frac{1}{2}}Z_t),\ \Pred_{[\underline{a},\overline{a}]}\otimes dt\text{-q.e.}.
\end{equation}
Hence since $K^P$ is non-decreasing, then combining (\ref{eq:hdomilambdalessh}) and (\ref{eq:ForFinVarDec}) imply that the finite variation part of the $Q$-semimartingale $R^{\bar{\phi}}$ is non-increasing.
Furthermore one has $R^{\bar{\phi}}\in\mathbb{D}^2(\filt^P,P)$ because $\pi^u_\cdot(X)\in\mathbb{D}^2\big(\filt^{\Pred_{[\underline{a},\overline{a}]}}\big)\subseteq \mathbb{D}^2(\filt^P,P)$ and $\bar{\phi}\in \Phi\subseteq \Phi(P)$. 
Now since $\lambda+\theta$ is bounded then $\frac{dQ}{dP}\big\rvert_{\F_T}$ is in $L^p(\F_T,P)$ for any $p<\infty$ and by H\"older's inequality it follows that $R^{\bar{\phi}}\in \mathbb{D}^{2-\epsilon}(\filt^P,Q)$  holds for some $\epsilon>0$.
Thus $R^{\bar{\phi}}$ is a $(\filt^P,Q)$-supermartingale.

We turn to the proof of the first claim of the theorem. To show that $\bar{\phi}$ solves the hedging problem (\ref{eq:GD-hedging-Problem}), 
let $P\in\Pred_{[\underline{a},\overline{a}]}$. Then by (\ref{eq:PiUStandard}) and $\Phi\subseteq \Phi(P')$ for all $P'\in \Pred_{[\underline{a},\overline{a}]}$ one has $P$-a.s.\ 
\begin{equation*}\begin{split}
 \pi^u_t(X)&\le \esssup^{\qquad\quad P}_{P'\in\Pred_{[\underline{a},\overline{a}]}(t,P,\filt_+)} \essinf^{\qquad\quad P}_{\phi\in\Phi}\rho^{P'}_t\Big( X-\int_t^T\phi_s^{\text{tr}}\big(\widehat{a}^{1/2}_s\widehat{\xi}_sds+dB_s\big)\Big)\\ 
           &\le  \essinf^{\qquad\quad P}_{\phi\in\Phi}\esssup^{\qquad\quad P}_{P'\in\Pred_{[\underline{a},\overline{a}]}(t,P,\filt_+)}\rho^{P'}_t\Big( X-\int_t^T\phi_s^{\text{tr}}\big(\widehat{a}^{1/2}_s\widehat{\xi}_sds+dB_s\big)\Big)\\
           &=\essinf^{\qquad\quad P}_{\phi\in\Phi}\rho_t(X-\int_t^T\phi^{\text{tr}}_s\big(\widehat{a}^{1/2}_s\widehat{\xi}_sds+dB_s\big)) \quad \text{ for $t\in[0,T]$.}
\end{split}
\end{equation*}
 To conclude that some $\bar{\phi} \in \Phi$ is a good-deal hedging strategy satisfying (\ref{eq:GD-hedging-Problem}),
it suffices to show for all $\theta\in\Theta$, $P\in \Pred_{[\underline{a},\overline{a}]}$ that $P$-a.s.\ for all $t\in[0,T]$, $P'\in \Pred_{[\underline{a},\overline{a}]}(t,P,\filt_+)$ and $Q\in \mathcal{P}^{\text{ngd}}(Q^{P',\theta})$ holds
\begin{equation*}
\pi^{u}_t(X) \ge E^Q_t\big[X-\int_t^T\bar{\phi}^{\text{tr}}_s\big(\widehat{a}^{1/2}_s\widehat{\xi}_sds+dB_s\big)\big].
\end{equation*}
To this end, let $\theta\in\Theta$, $P\in \Pred_{[\underline{a},\overline{a}]}$. By the first claim of the theorem, the tracking error $R^{\bar{\phi}}:=R^{\bar{\phi}}_\cdot(X)$ of 
$\bar{\phi}$ is a $(\filt^{P'},Q)$-supermartingale for every $Q\in \mathcal{P}^{\text{ngd}}(Q^{P',\theta})$, $P'\in \Pred_{[\underline{a},\overline{a}]}$. 
This implies by Lemma \ref{lem:QSAggregPS} that $\pi^u_t(X) - \pi^u_0(X)-\int_0^t\bar{\phi}_s^{\text{tr}}\big(\widehat{a}^{1/2}_s\widehat{\xi}_sds+dB_s\big)\ge E^Q_t\big[X - \pi^u_0(X)-\int_0^T\bar{\phi}_s^{\text{tr}}\big(\widehat{a}^{1/2}_s\widehat{\xi}_sds+dB_s\big)\big],\ P\text{-a.s.}$,
for all $Q\in \mathcal{P}^{\text{ngd}}(Q^{P',\theta}),$ $P'\in \Pred_{[\underline{a},\overline{a}]}(t,P,\filt_+)$, for any $t\le T$.
Reorganizing that inequality yields the claim.
\end{proof}
\begin{remark}\label{rem:ABoutAggregationandRObustness}
1.\ By a result of \citet{Karandikar} it is possible to define 
the stochastic integral $\int_0^\cdot Z_t^{\text{tr}}dB_t$ pathwise if the process $Z$ is c\`adl\`ag and $\filt$-adapted, such that in particular it satisfies the hypothesis of Theorem \ref{thm:GDHedgingTheorem}. 
Although the $Z$-component of a 2BSDE solution is not guaranteed in general to be c\`adl\`ag, 
we emphasize that Theorem \ref{thm:GDHedgingTheorem} may still be applicable in some specific situations. For instance in a  Markovian diffusion setting, one may be able to use partial differential equation (PDE) arguments 
to show that the $Z$-component is even continuous. An example in such a setting is provided in Section \ref{sec:ExamplewithPutOptionVolUncert}, where beyond the continuity of $Z$ we can even 
obtain explicit solutions to the 2BSDE (\ref{eq:2BSDEpiU}), for some bounded contingent claims. In the most general situation, however, the result of \citet{Nutz} can be used under additional set-theoretical assumptions 
to get rid of the aggregation condition in $\Phi$ and in the statement of Theorem \ref{thm:GDHedgingTheorem}; cf.\ Remarks \ref{rem:OnAggregationOfK} and \ref{rem:OnAggPhi}.

2.\ A consequence of Theorem \ref{thm:GDHedgingTheorem}  is a minmax identity: For $P\in\Pred_{[\underline{a},\overline{a}]}$ holds 
a.s. 
\begin{align*}
 \pi^u_t(X)&=\esssup^{\qquad\quad P}_{P'\in\Pred_{[\underline{a},\overline{a}]}(t,P,\filt_+)} \essinf^{\qquad\quad P}_{\phi\in\Phi(P')}  \rho^{P'}_t\Big( X-\phantom{}^{(P)}\hspace{-0.1cm}\int_t^T\phi_s^{\text{tr}}\big(\widehat{a}^{1/2}_s\widehat{\xi}_sds+dB_s\big)\Big)\\
		    &= \essinf^{\qquad\quad P}_{\phi\in\Phi} \esssup^{\qquad\quad P}_{P'\in\Pred_{[\underline{a},\overline{a}]}(t,P,\filt_+)} \rho^{P'}_t\Big( X-\int_t^T\phi_s^{\text{tr}}\big(\widehat{a}^{1/2}_s\widehat{\xi}_sds+dB_s\big)\Big), \quad \text{$t\in[0,T]$}.
\end{align*}
\end{remark}
\subsection{A case study: Hedging a put on non-traded but correlated asset}\label{sec:ExamplewithPutOptionVolUncert} 
In this subsection we investigate a simple Markovian example to provide more intuition and to illustrate the general but abstract  main theorems,
giving closed-form formulas for robust good-deal valuations and hedges. To this end, we investigate here
 the particular application of a vanilla put option  on a non-traded asset in concrete detail. 
The latter is a typical problem for optimal partial hedging in incomplete markets.
The example is elementary enough to permit even for closed-form solutions, in some parameter settings, and it can be solved 
also by more standard optimal control methods, exploiting the Markovian structure and certain properties of the claim's payoff function.
We  will make use of this fact to demonstrate clear differences to hedging by superreplication, and further to elucidate the difficulties arising from combined drift and volatility uncertainty, compared to only one type of ambiguity. 
More precisely, we show how under combined uncertainty the optimal control problem of robust good-deal valuation 
becomes effectively one over a non-rectangular domain of control variables, 
making it more complex to identify optimizers and worst-case priors (or parameters), even for monotone convex payoff functions for which intuition from examples in  other related valuation approaches (e.g.\ robust superreplication)
might at first suggest otherwise.  Of course, one cannot expect to get similarly elementary solutions for measurable contingent claims in general models: For the general problem formulation, 
the solution has been fully characterized by means of 2BSDEs in the main Theorems \ref{thm:2BSDECHaractforPiU} and \ref{thm:GDHedgingTheorem}. 

Let us consider a financial market where  only one risky asset (a stock) with discounted price process $S$ is tradeable, apart from the riskless asset (with unit price). 
In addition, there is another asset  that is not tradeable but whose value $L$ is correlated with $S$. The processes $S$ and $L$
are,  $\Pred_{[\underline{a},\overline{a}]}\text{-q.s.}$, given by 
\begin{equation*}
	dS_t = S_t(bdt+\sigma^SdB^1_t)\quad \text{and}\quad dL_t=L_t\big(\gamma dt+\beta(\rho dB^1_t+\sqrt{1-\rho^2}dB^2_t)\big),
\end{equation*}
where $B=(B^1,B^2)$ is the canonical process, $\Pred_{[\underline{a},\overline{a}]}$ is the set defined as in (\ref{eq:DefPH}) for diagonal matrices 
$\underline{a} =\mathrm{diag}(\underline{a}_1,\underline{a}_2)$ and $\overline{a} =\mathrm{diag}(\overline{a}_1,\overline{a}_2)$ in $\mathbb{S}_2^{>0}$, 
with $S_0,L_0,\sigma^S,\beta$ in $(0,\infty)$, for drifts $b,\gamma$ in $\R$,  volatility matrix $\sigma:=(\sigma^S,0)\in\R^{1\times 2}$ and a $P^0$-correlation coefficient $\rho\in[-1,1]$. For a constant bound $h\in[0,\infty)$ on the instantaneous Sharpe ratios, we are going to to derive closed-form expressions for robust good-deal valuations 
and hedges, for a European put option $X:=(\mathcal{K}-L_T)^+$ on the non-traded asset $L$ with strike $\mathcal{K}\in (0,\infty)$ and maturity $T$, and to identify the corresponding 
worst-case drifts and volatilities. At first we assume the drift rate $b$ of $S$ to be zero, so that the center market price of risk $\widehat{\xi}$ vanishes quasi-surely, and 
consider the case of uncertainty solely on volatility (i.e.\ for $\delta\equiv 0$). 
For this case we will identify  a worst-case volatility for 
the robust valuation bounds in closed form, and compare robust good-deal hedging with classical robust superreplication under volatility uncertainty. 
After this, we discuss the more complex case  with combined drift 
and volatility uncertainties and the difficulties that arise in identifying worst-case drifts and volatilities.
Finally, we investigate sensitivities of the derived robust good-deal bound under volatility uncertainty with respect to variations of the drift parameter $\gamma\in\R$ for the non-traded asset.

\subsubsection{Uncertainty solely about the volatility}\label{subsubsec:onlyVolUncert}
Denoting the entries of the processes $\widehat{a}$ and its square root $\widehat{a}^{\frac{1}{2}}$ by
\begin{equation*}\widehat{a} = \left(\begin{array}{cc}\widehat{a}^{11}&\widehat{a}^{12}\\\widehat{a}^{12}&\widehat{a}^{22}\end{array}\right)\quad \text{and}\quad  
\widehat{a}^{\frac{1}{2}} = \left(\begin{array}{cc}\widehat{c}^{11}&\widehat{c}^{12}\\\widehat{c}^{12}&\widehat{c}^{22}\end{array}\right),
\end{equation*}
one has  $\sigma \widehat{a}^{\frac{1}{2}} = \sigma^S(\widehat{c}^{11}, \widehat{c}^{12})$, 
\begin{equation}
\begin{split}\label{eq:Relationaandsqrta}
\widehat{a}^{11}&=(\widehat{c}^{11})^2+(\widehat{c}^{12})^2,\quad \widehat{a}^{12}=\widehat{c}^{12}(\widehat{c}^{11}+\widehat{c}^{22}),\quad \widehat{a}^{22}=(\widehat{c}^{22})^2+(\widehat{c}^{12})^2,\\
&\quad \text{and}\quad \widehat{a}^{11}\widehat{a}^{22}-(\widehat{a}^{12})^2 = \big(\widehat{c}^{11}\widehat{c}^{22}-(\widehat{c}^{12})^2\big)^2,
\end{split}
\end{equation}
implying $\mathrm{Im}\,(\sigma\widehat{a}^{\frac{1}{2}})^{\text{tr}} = \big\lbrace z\in\R^2 : \widehat{c}^{12}z_1-\widehat{c}^{11}z_2=0\big\rbrace$ and
$\mathrm{Ker}\,(\sigma\widehat{a}^{\frac{1}{2}}) = \big\lbrace z\in\R^2 : \widehat{c}^{11}z_1+\widehat{c}^{12}z_2=0\big\rbrace$. Hence for $z = (z_1,z_2)^{\text{tr}}\in\R^2$ one gets
\begin{equation}\label{eq:ExpressionPizPibotZ}
	\widehat{\Pi}(z) = \frac{1}{\widehat{a}^{11}}\left(\begin{array}{cc}(\widehat{c}^{11})^2z_1+\widehat{c}^{11}\widehat{c}^{12}z_2\\\widehat{c}^{11}\widehat{c}^{12}z_1+(\widehat{c}^{12})^2z_2\end{array}\right)\ \text{and}
		\ \widehat{\Pi}^\bot(z) = \frac{1}{\widehat{a}^{11}}\left(\begin{array}{cc}(\widehat{c}^{12})^2z_1-\widehat{c}^{11}\widehat{c}^{12}z_2\\(\widehat{c}^{11})^2z_2-\widehat{c}^{11}\widehat{c}^{12}z_1\end{array}\right). 
\end{equation}
In the case $\delta\equiv 0$, the 2BSDE (\ref{eq:2BSDEpiU}) thus rewrites here as 
\begin{equation}\label{eq:2BSDEpiUvolUncert}
 \begin{split}
 Y_t = X-\phantom{}^{(P)}\hspace{-0.1cm}\int_t^TZ_s^{\text{tr}}dB_s- \int_t^T \widehat{F}_s(\widehat{a}^{1/2}_sZ_s) ds+K^P_T-K^P_t,\ t\in[0,T],\ \Pred_{[\underline{a},\overline{a}]}\text{-q.s.}, 
 \end{split}
\end{equation}
where from (\ref{eq:ExpressionPizPibotZ}) and (\ref{eq:Relationaandsqrta}) one has , $ \Pred_{[\underline{a},\overline{a}]}\otimes dt\text{-q.e.}  $,
\begin{equation}\label{eq:SimplerGen}
\widehat{F}_t(\widehat{a}^{1/2}_tz)= -h \big\lvert\widehat{\Pi}^{\bot}_t\big(\widehat{a}_t^{1/2}z\big)\big\rvert= -h\big(\widehat{a}_t^{11}\widehat{a}_t^{22}-(\widehat{a}_t^{12})^2\big)^{1/2}\big(\widehat{a}_t^{11}\big)^{-1/2}\big\lvert z_2\big\rvert,
\end{equation}
for $z=(z_1,z_2)^{\text{tr}}\in\R^2$. Clearly $L_T$ is $\F_T$-measurable, and since the put option payoff function $x\mapsto (\mathcal{K}-x)^+$ is bounded and continuous
it follows that $X$ is $\F_T$-measurable and in $\mathbb{L}^2(\filt_+)$. Therefore the conditions of Theorem \ref{thm:2BSDECHaractforPiU} are satisfied, yielding that
the worst-case good-deal bound $\pi^u_\cdot(X)$ coincides with the $Y$-component of the solution of the 2BSDE (\ref{eq:2BSDEpiUvolUncert}).  
In Lemma \ref{lem:Explicit2BSDESotution} below we will express the solution to the 2BSDE (\ref{eq:2BSDEpiUvolUncert}) in terms of the 
classical solution $v\in\mathcal{C}^{1,2}\big([0,T)\times(0,\infty)\big)$ to the Black-Scholes' type PDE 
\begin{align}\label{eq:BSTypePDEvOLuNC}
	&\frac{\partial v}{\partial t} + \big(\gamma-h\beta\sqrt{1-\rho^2}\sqrt{\overline{a}_2}\big)\,x\,\frac{\partial v}{\partial x}+\frac{1}{2}\beta^2\,\big(\rho^2\overline{a}_1+(1-\rho^2)\overline{a}_2)\,x^2\,\frac{\partial^2v}{\partial x^2}=0\\ 
& \nonumber \quad 	\text{on the set $[0,T)\times (0,\infty)$, with boundary condition }\; v(T,\cdot) = (\mathcal{K}-\cdot)^+.
\end{align}
To this end, let $P^{\overline{a}} = P^0\circ(\overline{a}^{1/2}B)^{-1}\in\Pred_{[\underline{a},\overline{a}]}$ such that $\langle B\rangle_t = \overline{a}t\ P^{\overline{a}}\otimes dt\text{-a.s.}$. 
The process $L$ under $P^{\overline{a}}$ is a geometric Brownian motion with dynamics
\begin{equation*}
dL_t = L_t\Big(\gamma dt + \bar{\beta}\big(\bar{\rho}dW^{1,P^{\overline{a}}}_t+\sqrt{1-\bar{\rho}^2}\,dW^{2,P^{\overline{a}}}_t\big)\Big),\ t\in[0,T],\ P^{\overline{a}}\text{-a.s.},
\end{equation*}
where $W^{P^{\overline{a}}} = (W^{1,P^{\overline{a}}},W^{2,P^{\overline{a}}}):={(\overline{a})}^{\,-1/2}B$ is a $P^{\overline{a}}$-Brownian motion and 
\begin{equation*}
\bar{\beta}:=\beta\Big(\rho^2\overline{a}_1+(1-\rho^2)\overline{a}_2\Big)^{\frac 1 2}>0\;\text{and}\;
 \bar{\rho} := \rho\sqrt{\overline{a}_1}\Big(\rho^2\overline{a}_1+(1-\rho^2)\overline{a}_2\Big)^{-\frac 1 2}\in[-1,1].
\end{equation*}
Hence a closed-form expression for $v(t,L_t)$, with $v$  being the solution to the PDE (\ref{eq:BSTypePDEvOLuNC}), is given by the Black-Scholes formula for the price of the put option 
$X= (\mathcal{K}-L_T)^+$ in the model $P^{\overline{a}}$.
By arguments analogous to the derivations of the formulas in \citep[][Sect.3.2.1]{BechererKentiaTonleu}, $v(t,L_t)$ coincides $P^{\overline{a}}$-a.s.\ with the valuation bound 
$\pi^{u,P^{\overline{a}}}_t\big(X\big)$ for all $t\le T$, and is given in closed form as 
\begin{align}\label{eq:CloFormUpperGDBPutExVOlUnc}
	v(t,L_t)&=\pi^{u,P^{\overline{a}}}_t(X)
				  =\mathcal{K}N(-d_{-})- L_te^{m(T-t)}N(-d_{+})\\
		\nonumber    &= e^{m(T-t)}\ast\text{B/S-put-price}\big(\text{time: } t,\ \text{spot: }L_t,\ \text{strike: } \mathcal{K}e^{-m(T-t)}, \text{vol: } \bar{\beta} \big), 		    
\end{align}
where ``B/S-put-price'' denotes the standard Black-Scholes put pricing formula with zero interest rate, ``vol'' abbreviating volatility,  $N$ denoting the cdf of the standard normal law,
$
	m:= \gamma-h\bar{\beta}\sqrt{1-\bar{\rho}^2}= \gamma - h\beta\sqrt{1-\rho^2}\sqrt{\overline{a}_2},
$
 and $d_{\pm}:=\big(\ln \big(L_t/\mathcal{K}\big)+\big(m\pm \frac{1}{2}\bar{\beta}^2\big)(T-t)\big)\big(\bar{\beta}\sqrt{T-t}\big)^{-1}$.
The details of proof for the following lemma can be found in the appendix. 
\begin{lemma}\label{lem:Explicit2BSDESotution}
	The triple $(Y,Z,K)$ with $Y_t=v(t,L_t)$, $Z_t = \beta L_t\frac{\partial v}{\partial x}(t,L_t)\big(\rho,\sqrt{1-\rho^2}\big)^{\text{tr}}$,
and $K$ given by (\ref{eq:DefKAggreg}) for $v\in\mathcal{C}^{1,2}\big([0,T)\times(0,\infty)\big)$ solution to the PDE (\ref{eq:BSTypePDEvOLuNC})
	satisfies $(Y,Z)\in \mathbb{D}^2\big(\filt^{\Pred_{[\underline{a},\overline{a}]}}\big)\times \mathbb{H}^2\big(\filt^{\Pred_{[\underline{a},\overline{a}]}}\big)$ and 
	is the unique solution to the 2BSDE (\ref{eq:2BSDEpiUvolUncert}). In particular the stochastic integral $\int_0^\cdot Z_s^{\text{tr}}dB_s$ can be defined pathwise.
\end{lemma}
\paragraph{Worst-case model for valuation and hedging:}\label{worstcasepara}
Using Lemma \ref{lem:Explicit2BSDESotution}, Theorem \ref{thm:2BSDECHaractforPiU} implies by (\ref{eq:CloFormUpperGDBPutExVOlUnc}) that the robust good-deal bound
for a put option $X=(\mathcal{K}-L_T)^+$ is given in closed-form by 
\begin{equation}\label{eq:finalGDB-BlackScholes}
\pi^u_t(X) =  \mathcal{K}N(-d_{-})- L_te^{m(T-t)}N(-d_{+}),\ t\in[0,T],\ P^{\overline{a}}\text{-a.s.}. 
\end{equation}
Hence the robust good-deal bound $\pi^{u}_\cdot(X)$ for a put option $X= (\mathcal{K}-L_T)^+$ is attained at the ``maximal'' volatility matrix $\overline{a}$, and can be computed as in the absence of uncertainty, 
but in a worst-case model $P^{\overline{a}}\in\Pred_{[\underline{a},\overline{a}]}$ in which $\langle B\rangle_t = \overline{a}t$ holds $ {P^{\overline{a}}}\otimes dt$-a.e., yielding 
$\pi^{u}_t\big(X\big)=\pi^{u,P^{\overline{a}}}_t\big(X\big)$ for all $t\in[0,T],\ {P^{\overline{a}}}$-a.s..
In addition, $\pi^{u}_t\big(X\big)$ is given explicitly in terms of a Black-Scholes' type formula, for modified strike $\mathcal{K}\exp( -m(T-t))$ and volatility 
$\bar{\beta}= \beta\big(\rho^2\overline{a}_1+(1-\rho^2)\overline{a}_2\big)^{1/2}$.
Similarly Theorem \ref{thm:GDHedgingTheorem} and Lemma \ref{lem:Explicit2BSDESotution} imply by (\ref{eq:ExplicitFOrmforZ}) that the robust good-deal hedging strategy $\bar{\phi}:=\bar{\phi}(X)$ for $X$
is given by
\begin{align}
	\bar{\phi}_t &= -\beta e^{ m(T-t)}N(-d_+)L_t\,\widehat{a}^{-1/2}_t\,\widehat{\Pi}_t\Big(\widehat{a}^{1/2}_t\big(\rho,\sqrt{1-\rho^2}\big)^{\text{tr}}\Big)\notag\\
		      &= -\beta e^{ m(T-t)}N(-d_+)L_t\Big(\rho+\frac{\widehat{a}^{12}_t}{\widehat{a}^{11}_t}\sqrt{1-\rho^2}\, ,\ 0\Big)^{\text{tr}},
\quad \text{$P^{\overline{a}}\otimes dt$-a.e.},
  \label{eq:FinalPhiBarExamp}
\end{align}
where we have used the fact that $\widehat{a}^{-1/2}\widehat{\Pi}\big(\widehat{a}^{1/2}z\big) = \Big(z_1+\frac{\widehat{a}^{12}}{\widehat{a}^{11}}z_2\, ,\ 0\Big)^{\text{tr}}$ for $z=(z_1,z_2)^{\text{tr}}\in\R^2,$ 
which is straightforward by (\ref{eq:ExpressionPizPibotZ}) and (\ref{eq:Relationaandsqrta}).  Analogously the lower good-deal bound $\pi^{l}_\cdot\big(X\big)$ and corresponding hedging 
strategy can also be computed, but under the worst-case measure $P^{\underline{a}}\in\Pred_{[\underline{a},\overline{a}]}$ 
for ``minimal'' volatility matrix $\underline{a}$. 

\paragraph{Comparison with robust superreplication under volatility uncertainty:}
Intuitively as the bound $h$ on the Sharpe ratios increases to infinity, the (upper) good-deal bound $\pi^u_\cdot(X)$ should increase towards the robust upper no-arbitrage valuation bound, studied in \citet{AvellanedaetAl,Lyons95,DenisMartini,Vorbrink2014,NutzSoner,NeufeldNutz}. 
Our result in this example is in accordance with \citet{AvellanedaetAl,Lyons95,KarouiJeanBlancShreve98,EpsteinJi2013,Vorbrink2014} 
who showed that under ambiguous volatility, Black-Scholes-valuation and hedging of vanilla put (or call) options under maximal (resp.\ minimal) volatility corresponds to the worst-case for the seller (resp.\ for the buyer).
Instead of superreplication
we focus here on robust good-deal hedging under volatility uncertainty for valuation at the worst-case good-deal bound.
Beyond some simplifications (particular payoff function) in comparison to \citet{AvellanedaetAl,KarouiJeanBlancShreve98,Vorbrink2014}, our setup also includes some more original aspects: We study an option on a non-tradable asset $L$  with possibly non-perfect correlation with $S$, and there is incompleteness in the sense that the option $X$ is already  clearly non-replicable by dynamic trading  under any individual prior (if $|\rho|<1$), even without model uncertainty.
Further, we now show that the robust good-deal hedging strategy $\bar{\phi}(X)$ is very different from (the risky asset component of) the super-replicating strategy, in general.
Indeed since $0\le X\le \mathcal{K}$ holds pathwise, the (upper) no-arbitrage bound (or superreplication price)  process $\widehat{V}(X)$ 
under $P^{\overline{a}}$ defined by 
\begin{equation*}
\widehat{V}_t(X):=\esssup^{\qquad\quad P^{\overline{a}}}_{Q\in\mathcal{M}^e(P^{\overline{a}})}E^{Q}_t[X],\ t\in[0,T],\ P^{\overline{a}}\text{-a.s.}
\end{equation*}
satisfies $\pi^{u,P^{\overline{a}}}_t(X;h)\le \widehat{V}_t(X)\le \mathcal{K},\ P^{\overline{a}}\text{-a.s.},$ for $\pi^{u,P^{\overline{a}}}_t(X;h)$ denoting the good-deal bound in the 
model $P^{\overline{a}}$ with Sharpe ratio constraint $h\in[0,\infty)$ and being given by (\ref{eq:CloFormUpperGDBPutExVOlUnc}) for arbitrary but fixed $h$.
If $\lvert \rho\rvert<1$ then $\pi^{u,P^{\overline{a}}}_t(X;h)$ for $t<T$ increases to $\mathcal{K}$ as $h$ tends to $+\infty$ (since $m\to-\infty$, $d_{\pm}\to-\infty$),
and hence we get, in striking difference to the good deal bound from (\ref{eq:finalGDB-BlackScholes}), that
\begin{equation}\label{eq:VbarExplicit}
	\widehat{V}_t(X)=\mathcal{K}\mathds{1}_{\{t<T\}}+X\mathds{1}_{\{t=T\}},\ t\in[0,T],\ P^{\overline{a}}\text{-a.s.}.
\end{equation}
The superreplication price process $\widehat{V}(X)$ has the optional decomposition \citep[cf.][Thm.3.2]{Kramkov}
\begin{equation*}
\widehat{V}_t(X)=\widehat{V}_0(X)+\int_0^t\widehat{\phi}_sdB^1_s-\widehat{C}_t,\ t\in[0,T],\ P^{\overline{a}}\text{-a.s.}, 
\end{equation*}
where $\int_0^\cdot\widehat{\phi}_sdB^1_s$ and $\widehat{C}$ are unique \citep[see][Thm.2.1 and Lem.2.1]{Kramkov}. One obtains by (\ref{eq:VbarExplicit}) that $\int_0^\cdot\widehat{\phi}_sdB^1_s=0$ and $C=(\mathcal{K}-X)\mathds{1}_{\{T\}}$.
Note from (\ref{eq:FinalPhiBarExamp}) that $\bar{\phi}=(Z^1,0)^{\text{tr}}\ P^{\overline{a}}\otimes dt\text{-a.e.}$ since $\widehat{a}=\overline{a}\ P^{\overline{a}}\otimes dt\text{-a.e.}$.
For $\rho\neq 0$, the process $Z^1$ is non-trivial under $P^{\overline{a}}\otimes dt$. Overall for $0<\lvert \rho\rvert<1$, $\int_0^\cdot\bar{\phi}^{\text{tr}}_sdB_s=\int_0^\cdot Z^1_sdB^1_s$ cannot be equal to 
$\int_0^\cdot\widehat{\phi}_sdB^1_s\equiv0$ $P^{\overline{a}}\otimes dt$-almost everywhere. This means that for $0<\lvert \rho\rvert<1$
the good-deal hedging strategy $\bar{\phi}$ for the put option on the non-traded asset is not the traded asset component of the option's super-replicating strategy for the model $P^{\overline{a}}$. Therefore the robust good-deal hedging strategy 
$\bar{\phi}$ is in general not a super-replicating strategy under volatility uncertainty for the set $\Pred_{[\underline{a},\overline{a}]}$ of reference priors and 
$\pi^u_t(X)+\int_t^T\bar{\phi}^{\text{tr}}_s dB_s$ does not dominate the  claim $X$  $P^{\overline{a}}$-almost surely, let alone $\Pred_{[\underline{a},\overline{a}]}$-quasi-surely.

\subsubsection{Case with combined drift and volatility uncertainty}\label{subsubsec:CaseCombinedUncertExample}
The approach of Section \ref{subsubsec:onlyVolUncert} to derive closed-form valuations and hedges may not work in general when there is drift uncertainty in  
addition to volatility uncertainty. Indeed one would need, by the preceding valuation and hedging Theorems \ref{thm:2BSDECHaractforPiU} and \ref{thm:GDHedgingTheorem}, first to identify a candidate worst-case drift parameter $\bar{\theta}$ (and then possibly a worst-case volatility) 
satisfying (\ref{eq:OptEquaThetaBar}). This requires finding a minimizer over $\theta\in \Theta\equiv\{x\in\R^2: \lvert x\rvert\le \delta\}$ of the quantities $\widehat{F}^\theta (\widehat{a}^{1/2}Z)$ given by 
(\ref{eq:GeneratorFunctionFtheta}) for the $Z$-component of the solution to the 2BSDE (\ref{eq:2BSDEpiU}). However, it is not clear whatsoever how this could be done  in general (let alone explicitly), given that the expression for $\widehat{F}^\theta$
is quite complex by non-triviality of the kernel $\mathrm{Ker}\,(\sigma\widehat{a}^{1/2})$. Recall that the latter is due to the possible singularity of the volatility matrices $\sigma_t\in\R^{d\times n}$ when $d< n$, i.e.\  
under market incompleteness under each prior. Furthermore, even if one could 
identify $\bar{\theta}$, it would still be questionable what the corresponding worst-case volatility should be as $\widehat{F}^{\bar{\theta}}(\widehat{a}^{1/2}Z) = \widehat{F}(\widehat{a}^{1/2}Z)$
is still a very complicated function of the coefficients of $a\in\mathbb{S}^{>0}_2$. This issue does not appear in the less general setting of Section \ref{subsubsec:onlyVolUncert} where,
thanks to the zero drift assumption $b=0$ for the traded asset $S$, the expression for $\widehat{F}(\widehat{a}^{1/2}Z)$ greatly simplifies to (\ref{eq:SimplerGen}) and this allows by direct comparison 
to obtain $\overline{a}$ as the corresponding worst-case volatility. If market incompleteness is mainly due to the presence of volatility uncertainty (i.e.\ market is complete under every prior, and hence $\mathrm{Ker}\,(\sigma\widehat{a}^{1/2})$ is trivial),
then drift uncertainty is redundant as it does not have influence on (essentially superreplication) valuation bounds (and respective strategies) 
for any contingent claim, as $\widehat{F}^\theta(z) = \widehat{F}(z)= \widehat{\xi}^{\,\textrm{tr}}z$ for any $\theta\in\Theta,\ z\in\R^n$.
This has been argued in more detail by \citep[][Example 3]{EpsteinJi2013} and is perhaps not surprising. Indeed, it is well-known that the Black-Scholes price does not depend on the drift of the underlying and consequently remains unaffected 
by drift uncertainty. Super-replication (q.s.) appears as a natural  notion for robust hedging under market completeness for every prior and is well studied in the literature, \citep[cf.][]{AvellanedaetAl,Lyons95,NutzSoner,NeufeldNutz,EpsteinJi2013,Vorbrink2014},
where for the above reason  it is standard to assume zero drift. However in the case of incomplete markets (i.e.\ for $d<n$),
combined uncertainty on drifts and volatilities becomes  relevant for related approaches to valuation and hedging, that are
less expensive than superreplication. One might wonder whether one could identify worst-case drifts and volatilities explicitly for certain examples. Yet, even for the vanilla put 
option of Section \ref{subsubsec:onlyVolUncert}, we are not aware of a closed-form solution for this. E.g.\ for non-zero drift $b$, which is not even subject to uncertainty, we admit 
that we are not able to state a worst-case volatility in closed form  (e.g.\ by identifying it as $\bar{a}$ like we did on p.\pageref{worstcasepara}). 
To better demonstrate why this case is effectively 
already less tractable, let us assume for simplicity the following values for the model parameters: $\sigma^S=1,\ \gamma=0,\ \beta=1,\ \delta=0$ and $\rho =0$. The dynamics for $S$ and $L$ are then 
\begin{equation*}
	dS_t = S_t(bdt +dB^1_t)\quad \text{and}\quad dL_t = L_tdB^2_t,\quad \Pred_{[\underline{a},\overline{a}]}\text{-q.s.}.
\end{equation*}
Here $\widehat{\xi} = \big(b\widehat{c}^{11}/\widehat{a}^{11},b\widehat{c}^{12}/\widehat{a}^{11}\big)^{\text{tr}}$, $\mathrm{Ker}\,(\sigma\widehat{a}^{1/2}) = \text{Span}\big\lbrace \widehat{\eta}\big\rbrace$ with $\widehat{\eta} := \big(1,-\widehat{c}^{11}/\widehat{c}^{12}\big)^{\text{tr}}$.
By (\ref{eq:NGDMeasures}) for $P$ in $\Pred_{[\underline{a},\overline{a}]}$, the set $\mathcal{Q}^{\text{ngd}}(P)$ consists of all measures $Q^\varepsilon$ with $dQ^\varepsilon=\phantom{}^{(P)}\hspace{-0.05cm}\mathcal{E}\big((-\widehat{\xi}+\varepsilon\widehat{\eta})\cdot W^P\big)dP,\ P$-a.s.,
for $\mathbb{R}$-valued processes $\varepsilon$ satisfying $\lvert\varepsilon\rvert\le \widehat{c}^{12}(\widehat{a}^{11})^{-1/2}\big(h^2-b^2/\widehat{a}^{11}\big)^{1/2}$ everywhere.
Since $B^2=(\widehat{c}^{12},\widehat{c}^{22})\cdot W^P$ and $(\widehat{c}^{12})^2+(\widehat{c}^{22})^2=\widehat{a}^{22}$, then a change of measure from $P$ to $Q^\varepsilon$ yields 
\begin{equation*}
	dL_t = L_t \left(\widehat{a}^{22}_tdW^\varepsilon_t -\left( b\,\widehat{a}^{12}_t\,(\widehat{a}^{11}_t)^{-1}+\varepsilon_t(\widehat{c}^{12}_t)^{-1}\big(\widehat{a}^{11}_t\widehat{a}^{22}_t-(\widehat{a}^{12}_t)^2\big)^{1/2}\right)dt\right)\quad Q^\varepsilon\text{-a.s.},
\end{equation*}
where $W^\varepsilon$ is a one-dimensional $Q^\varepsilon$-Brownian motion. As the put payoff function $(\mathcal{K}-\cdot)^+$ is non-increasing, the good-deal bound 
under $P$ at time zero $\pi^{u,P}_0((\mathcal{K}-L_T)^+) = \sup_{\varepsilon} E^{Q^\varepsilon}[(\mathcal{K}-L_T)^+]$ is then attained at the largest possible $\varepsilon$ which is $\widehat{c}^{12}(\widehat{a}^{11})^{-1/2}\big(h^2-b^2/\widehat{a}^{11}\big)^{1/2}$. 
After replacing in the dynamics of $L$ and taking the supremum over all $P\in \Pred_{[\underline{a},\overline{a}]}$ we obtain by Part~2 of Remark \ref{rem:DynProgPrinciple} that the robust good-deal bound at time $t=0$ is 
\begin{equation}\label{eq:OptControlPiuMarkovian}
	\pi^{u}_0(X) = \sup_{a\in\mathbb{S}^{>0}_{2}\cap [\underline{a},\overline{a}]} E^{P^0}\big[\big(\mathcal{K}-L^{\gamma(a),\,\beta(a)}_T\big)^+\big],
\end{equation}
with 
\(
	\beta(a) := a^{22}\) and  \( \gamma(a) :=  -b\,(a^{12}/a^{11})-\big(h^2-b^2/a^{11}\big)^{\frac{1}{2}}\big(a^{11} a^{22}-(a^{12})^2\big)^{\frac{1}{2}}(a^{11})^{-\frac{1}{2}}.
\)
Moreover $L^{\gamma,\beta}$ has dynamics $dL^{\gamma,\beta}_t=L^{\gamma,\beta}_t\big(\gamma_t dt+\beta_t dW_t)\big)$, $t\in[0,T]$, for some one-dimensional $P^0$-Brownian motion $W$.
Hence, we recognize that (\ref{eq:OptControlPiuMarkovian}) is given by the value of a standard (Markovian) optimal control problem with state process $L^{\gamma,\beta}$ but with (dynamic) control variables $(\gamma_t,\beta_t)_{t\in [0,T]}$ taking values in a 
generally non-rectangular set 
\begin{equation*}
	 \big\lbrace (\gamma,\beta)\in\R^2\ \big \lvert\big.\ \beta=\beta(a) \text{ and } \gamma = \gamma(a) \text{ for some }a\in \mathbb{S}^{>0}_{2}\cap[\underline{a},\overline{a}] \big\rbrace,
\end{equation*}
what clearly is a complication for the derivation of optimal controls. Because the payoff function $x\mapsto (\mathcal{K}-x)^+$ is non-increasing and convex, then 
clearly a volatility matrix $a^*$ which simultaneously maximizes $a\mapsto \beta(a)$ and minimizes $a\mapsto \gamma(a)$ would be an optimizer for (\ref{eq:OptControlPiuMarkovian}).
In the particular case where $b=0$,  one can  easily check that $a^*$ equals $\overline{a}$, yielding (again) the simple closed-form solution 
of Section~\ref{subsubsec:onlyVolUncert}. Yet, when $b\neq 0$, it becomes less simple to describe $a^*$.

\subsubsection{Sensitivity with respect to drift parameter of non-traded asset} 

Despite the lack of closed-form expression for valuations under combined uncertainties, still, some further insight can be obtained by investigating simply the sensitivity of the robust good-deal bound 
$\pi^u_\cdot(X)=:\pi^u_\cdot(X;\gamma)$ with respect to a  (fixed constant) drift $\gamma$ of the non-traded asset $L$, solely under volatility uncertainty. This is  straightforward, by exploiting the 
explicit formulas obtained in Section \ref{subsubsec:onlyVolUncert}.
By (\ref{eq:finalGDB-BlackScholes}) one gets for any $t\in[0,T]$ and $\gamma\in\R$ that $\pi^u_t(X;\gamma) = e^{(T-t)m(\gamma)}\pi^{BS}_t(X;\gamma)$, $P^{\overline{a}}\text{-a.s.}$, where $\pi^{BS}_t(X;\gamma)$ denotes 
the Black-Scholes price at time $t$ for the put option $X=(\mathcal{K}-L_T)^+$ on underlying $L$ with volatility $\bar{\beta}$ under $P^{\overline{a}}$, for risk-free rate $m(\gamma):=\gamma - h\beta\sqrt{1-\rho^2}\overline{a}_2$. Hence differentiating 
$\pi^u_t(X): \gamma\mapsto \pi^u_t(X;\gamma)$ gives
\begin{equation*}
\frac{\partial \pi^u_t}{\partial \gamma} = (T-t)e^{(T-t)m(\gamma)} \pi^{BS}_t(\gamma)-\mathcal{K}(T-t)N(-d_-)=-(T-t)e^{(T-t)m(\gamma)}L_tN(-d_-).
\end{equation*}
Since the far right-hand side is always non-positive, then $\frac{\partial \pi^u_t(X)}{\partial \gamma}\le 0$. Therefore we obtain,  what intuition suggests: The robust good-deal bound $\pi^u_\cdot(X;\gamma)$ for the 
European put option $X=(\mathcal{K}-L_T)^+$ is non-increasing in $\gamma$ for the model $P^{\overline{a}}$. This implies that for $\underline{\gamma},\overline{\gamma}\in\R$ specifying an interval range $[\underline{\gamma},\overline{\gamma}]$ for the (constant) drift parameter of $L$, the worst-case drift corresponds to $\underline{\gamma}$, and is that for which the supremum $\esssup^{P^{\overline{a}}}_{\gamma\in[\underline{\gamma},\overline{\gamma}]}\pi^u_t(X;\gamma)$ is 
attained for any $t\in[0,T]$.  Note that this supremum may be different from the worst-case good-deal bound (defined in (\ref{eq:DefGDBounds}) and characterized by Theorem \ref{thm:2BSDECHaractforPiU}) under combined drift and volatility uncertainties, 
as the latter parametrizes drift uncertainty instead in terms of stochastic drifts $\gamma$ that satisfy $\underline{\gamma}\le \gamma\le \overline{\gamma}$ pointwise and there is no apparent reason why the worst-case 
volatility for every fixed (stochastic) drift $\gamma$ should be $\overline{a}$. 
 
\section{Appendix}
This section contains some proofs and details omitted in the main body of the paper. 
This includes the Lemma \ref{lem:SaddePointOptPb} which is used 
in the proof of Theorem \ref{thm:GDHedgingTheorem}.
\begin{proof}[Proof of Proposition \ref{pro:ExistenceUniquenessSol2BSDEs}]
	Consider the family $(\Pred(t,\omega))_{(t,\omega)\in[0,T]\times\Omega}$ of sets probability measures given by $\Pred(t,\omega):= \Pred_{[\underline{a},\overline{a}]}$ for all 
	$(t,\omega)\in[0,T]\times\Omega$ and the constant set-valued process $\mathrm{D}:[0,T]\times\Omega\to \mathbb{S}^{>0}_n$ with $\mathrm{D}_t(\omega):=[\underline{a},\overline{a}]$ for all $(t,\omega)\in[0,T]\times\Omega$. 
	Clearly $\mathrm{D}$ has the properties required in \citep[][Example 2.1]{NeufeldNutz} from which the first claim of \citep[][Cor.2.6]{NeufeldNutz} implies that the constant family 
	$\Pred(t,\omega)\equiv \Pred_{[\underline{a},\overline{a}]}$ satisfies Condition A therein, hence in particular the measurability and stability conditions of 
	Assumption 2.1, (iii)-(v) of \citet{PossamaiTanZhou}, whereby (iii) in particular follows from  \citep[][Condition (A1)]{NeufeldNutz} since a countable product of analytic sets is 
	again analytic \citep[][Prop.7.38]{BertsekasShreve}.
	This together with our Assumption \ref{asp:Assumption1} and $X\in \mathbb{L}^2(\filt_+)$ imply Assumptions 2.1 and 3.1 of \citet{PossamaiTanZhou} from which 
	a direct application of Theorems 4.1 and 4.2 therein yields existence and uniqueness of a 2BSDE solution $(Y,Z,(K^P)_{P\in\Pred_{[\underline{a},\overline{a}]}})\in \mathbb{D}^2\big(\filt^{\Pred_{[\underline{a},\overline{a}]}}_+\big)\times \mathbb{H}^2\big(\filt^{\Pred_{[\underline{a},\overline{a}]}}\big)\times \mathbb{I}^2\big(\big(\filt^P_+\big)_{P\in\Pred_{[\underline{a},\overline{a}]}}\big)$
satisfying the representation (\ref{eq:RepSol2BSDEs}), where $\filt^{\Pred_{[\underline{a},\overline{a}]}}_+=\big(\F_{t+}^{\Pred_{[\underline{a},\overline{a}]}}\big)_{t\in[0,T]}$, with $\F_{t+}^{\Pred_{[\underline{a},\overline{a}]}}:=\bigcap_{P\in\Pred_{[\underline{a},\overline{a}]}}\F^P_{t+}.$
	Note that the additional orthogonal martingale components in the 2BSDE formulation of \citet{PossamaiTanZhou} vanishes in this case thanks to the martingale representation property 
	of $\Pred_{[\underline{a},\overline{a}]}$ in Lemma \ref{lem:QSAggregPS}. Moreover since by Lemma \ref{lem:QSAggregPS} the filtration $\filt^P$ is actually right continuous for every $P\in \Pred_{[\underline{a},\overline{a}]}$, it follows that 
	$\F^P_{t+} = \F^P_{t}$ for $t\in[0,T],\ P\in \Pred_{[\underline{a},\overline{a}]}$ which implies that the minimum condition in the definition of a 2BSDE solution in \citep[][Def.4.1]{PossamaiTanZhou} is equivalent to (\ref{eq:MinCond}).
	In particular, we have $\filt^{\Pred_{[\underline{a},\overline{a}]}}_+ = \filt^{\Pred_{[\underline{a},\overline{a}]}}$. 
\end{proof}
\begin{proof}[Proof of Proposition \ref{Pro:RepresentationRhoX}]
By the classical comparison theorem for standard BSDEs, one easily sees for every $P\in\Pred_{[\underline{a},\overline{a}]}$ that $\mathcal{\widetilde{Y}}^{P,X} = \rho^{P}_\cdot(X),$ 
where $(\mathcal{\widetilde{Y}}^{P,X},\mathcal{\widetilde{Z}}^{P,X})$ is the unique solution to the standard Lipschitz BSDE under $P$ with data $(-\widehat{F}'(\widehat{a}^{\frac{1}{2}}\cdot),X)$
 for $F'$ given by (\ref{eq:DefforRho}). In addition, one can verify as in the proof of Theorem \ref{thm:2BSDECHaractforPiU} for $(F,X)$ that $(F',X)$ satisfies Assumption \ref{asp:Assumption1} as well.
The required result then follows from an application of Proposition \ref{pro:ExistenceUniquenessSol2BSDEs}.
\end{proof}
\begin{proof}[Proof of Lemma \ref{lem:Explicit2BSDESotution}]
By Theorem \ref{thm:2BSDECHaractforPiU} the 2BSDE (\ref{eq:2BSDEpiUvolUncert}) admits a unique solution which remains to be identified as claimed. For any $P\in\Pred_{[\underline{a},\overline{a}]}$, 
It\^o's formula and (\ref{eq:BSTypePDEvOLuNC}) yield for $t\in[0,T]$, $\ P\text{-a.s.}$, that 
\begin{equation*}
  v(t,L_t) = X-\int_t^TZ_s^{\text{tr}}dB_s + h\int_t^T \big(\widehat{a}_s^{11}\widehat{a}_s^{22}-(\widehat{a}_s^{12})^2\big)^{1/2}\big(\widehat{a}_s^{11}\big)^{-1/2}\big\lvert Z^2_s\big\rvert ds +K_T-K_t, 
  \end{equation*}
where, by using (\ref{eq:CloFormUpperGDBPutExVOlUnc}), the processes  $Z = (Z^1,Z^2)^{\text{tr}}$ and $K$ are given by 
\begin{align}
\label{eq:ExplicitFOrmforZ}
Z_t &=\beta L_t\frac{\partial v}{\partial x}(t,L_t)\big(\rho,\sqrt{1-\rho^2}\big)^{\text{tr}} = -\beta e^{m(T-t)}N(-d_+)L_t\big(\rho,\sqrt{1-\rho^2}\big)^{\text{tr}},\\
K_t &=  \int_0^t \bigg[h\beta\sqrt{1-\rho^2}L_s\frac{\partial v}{\partial x}(s,L_s)\Big(\big(\widehat{a}_s^{11}\widehat{a}_s^{22}-(\widehat{a}_s^{12})^2\big)^{1/2}\big(\widehat{a}_s^{11}\big)^{-1/2}-\sqrt{\overline{a}_2}\Big)\label{eq:DefKAggreg}\\
\nonumber
	&\;+\frac{1}{2}\beta^2L^2_s\frac{\partial^2 v}{\partial x^2}(s,L_s)\Big(\rho^2(\overline{a}_1-\widehat{a}^{11}_s)+(1-\rho^2)(\overline{a}_2-\widehat{a}^{22}_s)-2\rho\sqrt{1-\rho}\widehat{a}^{12}_s\Big)\bigg] ds.
\end{align}
To show that $K$ is a non-decreasing process, notice that $\widehat{a}\le \overline{a}\ P\otimes dt\text{-a.e.}$ yields $\widehat{a}^{1/2}\le \overline{a}^{1/2}\ P\otimes dt\text{-a.e.}$ and both inequalities imply that 
$P\otimes dt$-a.e.\ $\big(\widehat{a}^{11}\widehat{a}^{22}-(\widehat{a}^{12})^2\big)^{1/2}\big(\widehat{a}^{11}\big)^{-1/2}\le \big(\widehat{a}^{22}\big)^{1/2}\le \sqrt{\overline{a}_2}$ and 
\begin{equation*}
\begin{split}
	\rho^2(\overline{a}_1&\,-\widehat{a}^{11})+(1-\rho^2)(\overline{a}_2-\widehat{a}^{22})-2\rho\sqrt{1-\rho}\,\widehat{a}^{12}\\
	&=\big(\rho,\sqrt{1-\rho^2}\big)\overline{a}\,\big(\rho,\sqrt{1-\rho^2}\big)^{\text{tr}} - \big(\rho,\sqrt{1-\rho^2}\big)\widehat{a}\,\big(\rho,\sqrt{1-\rho^2}\big)^{\text{tr}}\ge 0.
\end{split}
\end{equation*}
Thus the process $K$ is $P$-a.s.\ non-decreasing because the delta of the put option in the Black-Scholes model is non-positive and the gamma is non-negative, i.e.\ $\frac{\partial v}{\partial x}(t,L_t)\le 0$ and $\frac{\partial^2 v}{\partial x^2}(t,L_t)\ge 0$
for all $t\in[0,T]$ using (\ref{eq:CloFormUpperGDBPutExVOlUnc}). Moreover it can be shown following arguments analogous to those in the proof of \citep[][Thm.5.3]{SonerTouziZhang-Wellposedness}
that the process $K$ satisfies the minimum condition (\ref{eq:MinCond}); we omit the essentially technical details which we refer to \citep[][Sect.4.3.3, proof of Lem.4.25]{KentiaPhD}. 

It remains to show that $v(\cdot,L_\cdot)\in\mathbb{D}^2\big(\filt^{\Pred_{[\underline{a},\overline{a}]}}\big)$, $Z\in \mathbb{H}^2\big(\filt^{\Pred_{[\underline{a},\overline{a}]}}\big)$ and that the stochastic integral 
$\int_0^\cdot Z_s^{\text{tr}}dB_s$ can be constructed pathwise. This will conclude by uniqueness of the solution to the 2BSDE that $(v(\cdot,L_\cdot),Z,K)$  is the unique solution to the 2BSDE (\ref{eq:2BSDEpiUvolUncert}) 
for $Z$ given as in (\ref{eq:ExplicitFOrmforZ}) and $K$ as in (\ref{eq:DefKAggreg}).  
Since $v$ is $\mathcal{C}^{1,2}$ and $L$ is $\Pred_{[\underline{a},\overline{a}]}$-q.s.\ continuous and $\filt^{\Pred_{[\underline{a},\overline{a}]}}$-adapted, then $v(\cdot,L_\cdot)$ is $\Pred_{[\underline{a},\overline{a}]}$-q.s.\ c\`ad\`ag 
and $\filt^{\Pred_{[\underline{a},\overline{a}]}}$-progressive and $Z$ is $\filt^{\Pred_{[\underline{a},\overline{a}]}}$-predictable. That $v(\cdot,L_\cdot)$ is in $\mathbb{D}^2\big(\filt^{\Pred_{[\underline{a},\overline{a}]}}\big)$
now follows from (\ref{eq:CloFormUpperGDBPutExVOlUnc}) which indeed implies $0\le v(t,L_t)\le \mathcal{K}$ pathwise. 
By (\ref{eq:ExplicitFOrmforZ}) and since $\underline{a}\le \widehat{a}\le \overline{a}$ holds $P\otimes dt$-a.e.\ for any $P\in\Pred_{[\underline{a},\overline{a}]}$, one has $\big\lvert \widehat{a}^{1/2}_tZ_t\big\rvert^2 \le \max(\overline{a}_1,\overline{a}_2)\beta^2e^{2\lvert m\rvert T} L^2_t $
$P\otimes dt$-a.e.\ for any $P\in\Pred_{[\underline{a},\overline{a}]}$. Hence to conclude that $Z\in \mathbb{H}^2\big(\filt^{\Pred_{[\underline{a},\overline{a}]}}\big)$ it suffices to show that $\sup_{P\in\Pred_{[\underline{a},\overline{a}]}}E^P\big[\int_0^TL^2_tdt\big]<\infty$. 
To this end, note that for any $P\in\Pred_{[\underline{a},\overline{a}]}$ hold
\begin{equation}\label{eq:IntermediaryforL}
\int_0^TL^2_tdt\le \beta^{-2}(\min(\underline{a}_1,\underline{a}_2))^{-1}\langle L\rangle_T\quad\text{and}\quad
L^2_T \le L^2_0e^{\big(2\lvert\gamma\rvert+\beta^2\max(\overline{a}_1,\overline{a}_2)\big)T}\tilde{L}_T
\end{equation}
$P$-almost surely, for $\tilde{L} = 1+\int_0^\cdot2\tilde{L}_s\beta\big(\rho dB^1_s+\sqrt{1-\rho^2}dB^2_s\big)\ \Pred_{[\underline{a},\overline{a}]}$-q.s.. Clearly 
$E^P[\tilde{L}_T]\le 1$ for every $P\in\Pred_{[\underline{a},\overline{a}]}$. 
Thus taking expectations in (\ref{eq:IntermediaryforL}) gives 
\begin{equation*}
E^P\Big[\int_0^TL^2_tdt\Big] \le \beta^{-2}(\min(\underline{a}_1,\underline{a}_2))^{-1}L^2_0e^{\big(2\lvert\gamma\rvert+\beta^2\max(\overline{a}_1,\overline{a}_2)\big)T},\quad \text{for all }P\in\Pred_{[\underline{a},\overline{a}]}.
\end{equation*}
Now taking the supremum over all $P\in\Pred_{[\underline{a},\overline{a}]}$ implies that $Z\in \mathbb{H}^2\big(\filt^{\Pred_{[\underline{a},\overline{a}]}}\big)$. As a consequence $(v(\cdot,L_\cdot),Z,K)$ is the unique solution to the 2BSDE (\ref{eq:2BSDEpiUvolUncert}).
Finally $\int_0^\cdot Z_s^{\text{tr}}dB_s$ can be constructed pathwise by \citet{Karandikar}, since $Z$ is  continuous and $\filt$-adapted. 
\end{proof}

\begin{lemma}\label{lem:SaddePointOptPb}
	For $d< n$, let $\sigma \in \R^{d\times n}$ be of full (maximal) rank $d$, $h>0$, $z\in\R^n$, $\xi\in C:=\mathrm{Im}\,\sigma^{\textrm{tr}}$, and
	$ U\subset \R^n$ be a convex-compact set. 
	Let $\Pi(\cdot)$ and $\Pi^\bot(\cdot)$ denote the orthogonal projections onto the linear subspaces $C$ and $C^\bot=\mathrm{Ker}\,\sigma$, respectively, and let $F: \R^n\times\R^n\to \R$
	with $F((\phi,\vartheta)= {\xi}^\textrm{tr}\phi -\vartheta^\textrm{tr}(z-\phi) - h\lvert z-\phi\rvert$.
	Assume that $\abs{\xi+\Pi(\vartheta)}<h$ for all $\vartheta\in U$. Then:
\\
1.\  the vector 
 \(
 \bar{\phi}(\vartheta) :=  \Pi(z)+\big\lvert\Pi^\bot(z)\big\rvert\Big(h^2-\lvert\xi+\Pi(\vartheta)\rvert^2\Big)^{-1/2}\big(\xi+\Pi(\vartheta)\big)
 \)
 is, for any $\vartheta\in\R^n$, the unique maximizer of $\phi\mapsto F(\phi,\vartheta)$ over $C$ , the maximum being
 \begin{equation*}
 G(\vartheta):=F(\bar{\phi}(\vartheta),\vartheta)=-\Pi^\bot(\vartheta)^{\text{tr}}\,\Pi^\bot(z)+\xi^{\text{tr}}\Pi(z) -\Big(h^2-\lvert\xi+\Pi(\vartheta)\rvert^2\Big)^{1/2}\big\lvert \Pi^\bot(z)\big\rvert.
 \end{equation*}
 \\
2.\  The minmax identity 
\[
\inf_{\vartheta\in U}\ \sup_{\phi\in C}F(\phi,\vartheta)= F(\bar{\phi}(\bar{\vartheta}),\bar{\vartheta})=G(\bar{\vartheta})=\sup_{\phi\in C}\ \inf_{\vartheta\in U}F(\phi,\vartheta)
\]
holds, for $\bar{\phi}(\bar{\vartheta})$ being the $\phi$-component of the saddle point with $\bar{\vartheta}=\argmin_{\vartheta\in U} G(\vartheta)$.
\\
3.\ Assume $0\in U$, then for $\bar{\vartheta}$ and $\bar{\phi}(\bar{\vartheta})$
 we have $F(\bar{\phi}(\bar{\vartheta}),\bar{\vartheta}) = \inf_{\vartheta\in U} F(\bar{\phi}(\bar{\vartheta}),\vartheta)$.
\end{lemma}
\begin{proof} As the proof of part 1\ is analogous to that of \citep[][Lem.6.1]{Becherer-Good-Deals} \citep[or, in a more general ellipsoidal setup, of ][Lem.5.1]{BechererKentiaTonleu},
we leave details to reader and just show parts 2 and 3 here. 

Part 2: For every $\phi\in\R^n$, the function $\vartheta\mapsto F(\phi,\vartheta)$ is convex, continuous. For every $\vartheta\in\R^n$ the function $\phi\mapsto F(\phi,\vartheta)$ is concave,
 continuous. As $ U\subset\R^n$ is convex and compact, and  $C$ is convex and closed,  a minmax theorem \citep[][Ch.VI, Prop.2.3]{EkelandTemam}
 applies and the minmax identity holds. Furthermore for any $\vartheta\in U$, the function $\phi\mapsto F(\phi,\vartheta)$ 
is strictly concave over $\{\Pi^\bot(\phi)=0\}$ if $\Pi^\bot(z)\neq 0$, and strictly concave at $\phi=z$ if $\Pi^\bot(z)= 0$, since $\abs{\xi+\Pi(\vartheta)}<h$. 
Hence \citep[][Ch.VI, Prop.1.5]{EkelandTemam} implies that the $\phi$-components of the saddle points are identical, in particular, to $\bar{\phi}(\bar{\vartheta})$ since indeed $(\bar{\phi}(\bar{\vartheta}),\bar{\vartheta})$
is a saddle point.

Part 3: The function $\R^n\ni\phi\mapsto \inf_{\vartheta\in U}F(\phi,\vartheta) = {\xi}^\textrm{tr}\phi -\sup_{\vartheta\in U}\vartheta^\textrm{tr}(z-\phi) - h\lvert z-\phi\rvert$ is concave and continuous. 
In addition this function is also coercive on $C$, i.e.\ $F(\phi)\to -\infty$ as $\abs{\phi}\to +\infty$ for $\Pi^\bot(\phi)=0$ because $\abs{\xi}<h$ and $\sup_{\vartheta\in U}\vartheta^{\text{tr}}(z-\phi)\ge0$
since $0\in U$. Hence by \citep[][Ch.II, Prop.1.2]{EkelandTemam} there exists $\phi^*\in C$ such that $\sup_{\phi\in C}\inf_{\vartheta\in U}F(\phi,\vartheta) = \inf_{\vartheta\in U}F(\phi^*,\vartheta)$.
In other words, $\phi^*$ is the $\phi$-component of a saddle point of $F$, with the other component being $\vartheta^* = \argmax_{\vartheta\in U} \vartheta^\textrm{tr}(z-\phi^*)$. By Part 2.,$\phi^* = \bar{\phi}(\bar{\vartheta})$,
and hence claim 3.\ follows.
\end{proof}
\bibliographystyle{abbrvnat}
\bibliography{Paper_BechererKentia}
\end{document}